%% file: main.tex
\documentclass[conference]{IEEEtran}
\IEEEoverridecommandlockouts
\usepackage{cite}
\usepackage{amsmath,amssymb,amsfonts}
\usepackage{latexsym}
\usepackage{comment}
\usepackage{algorithm}
\usepackage{graphicx}
\usepackage{textcomp}
\usepackage{xcolor}
\usepackage{mathabx}
\usepackage{stmaryrd}
\usepackage{algpseudocode}
\usepackage{multirow}

\usepackage{caption}
\def\BibTeX{{\rm B\kern-.05em{\sc i\kern-.025em b}\kern-.08em
    T\kern-.1667em\lower.7ex\hbox{E}\kern-.125emX}}

\begin{document}
\include{macro}

\title{Bayesian Statistical Model Checking for Multi-agent Systems using HyperPCTL*}

\author{\IEEEauthorblockN{Spandan Das}
\IEEEauthorblockA{\textit{Department of Computer Science} \\
\textit{Kansas State University}\\
Manhattan, Kansas \\
spandan@ksu.edu}
\and
\IEEEauthorblockN{Pavithra Prabhakar}
\IEEEauthorblockA{\textit{Department of Computer Science} \\
\textit{Kansas State University}\\
Manhattan, USA\\
pprabhakar@ksu.edu}
}

\maketitle

\begin{abstract}
In this paper, we present a Bayesian method for statistical model checking ($\smc$) 
of probabilistic hyperproperties specified in the logic $\hyppctl$ on discrete-time Markov chains ($\dtmc$s). 
While $\smc$ of $\hyppctl$ using sequential probability ratio test ($\sprt$) has been explored before, 
we develop an alternative $\smc$ algorithm based on Bayesian hypothesis testing.
In comparison to PCTL*, verifying $\hyppctl$ formulae is complex owing to their 
simultaneous interpretation on multiple paths of the $\dtmc$.
In addition, extending the bottom-up model-checking algorithm of the non-probabilistic 
setting is not straight forward due to the fact that $\smc$ does not return exact
answers to the satisfiability problems of subformulae, 
instead, it only returns correct answers with high-confidence. 
We propose a recursive algorithm for $\smc$ of $\hyppctl$ based on a modified 
Bayes' test that factors in the uncertainty in the recursive satisfiability results.
We have implemented our algorithm in a Python toolbox, HyProVer, and compared our approach with the $\sprt$ based $\smc$. 
Our experimental evaluation demonstrates that our Bayesian $\smc$ algorithm performs 
better both in terms of the verification time and the number of 
samples required to deduce satisfiability  of a given $\hyppctl$ formula.
\end{abstract}


\input{intro}
\input{related}
\input{Prelim}

\input{Problem}
\input{case}
\input{Hyp}
\input{algo}

\input{experiment}
\input{conclusion}

\section*{Acknowledgment}
This work was partially supported by NSF CAREER Grant No. 1552668 and NSF Grant No. 2008957.

\bibliography{SMC}
\bibliographystyle{IEEEtran}

\end{document}

%% file: macro.tex
\newtheorem{theorem}{Theorem}
\newtheorem{fact}[theorem]{Fact}
\newtheorem{proposition}[theorem]{Proposition}
\newtheorem{observation}[theorem]{Observation}
\newtheorem{lemma}[theorem]{Lemma}
\newtheorem{definition}[theorem]{Definition}
\newtheorem{corollary}[theorem]{Corollary}
\newtheorem{remark}[theorem]{Remark}
\newtheorem{claim}[theorem]{Claim}
\newtheorem{conjecture}[theorem]{Conjecture}
\newtheorem{assumption}[theorem]{Assumption}
\newtheorem{property}[theorem]{Property}

\newcommand{\qed}{\hfill \ensuremath{\Box}}

\newenvironment{proof}{
	\vspace*{-\parskip}\noindent\textit{Proof.}}{$\qed$
	
	\medskip
}

\newcommand{\alg}[1]{\mathsf{#1}}
\newcommand{\Prover}{\alg{P}}
\newcommand{\Verifier}{\alg{V}}
\newcommand{\Simulator}{\alg{S}}
\newcommand{\PPT}{\alg{PPT}}
\newcommand{\isom}{\cong}
\newcommand{\from}{\stackrel{\scriptstyle R}{\leftarrow}}
\newcommand{\handout}[5]{
	\noindent
	\begin{center}
		\framebox{
			\vbox{
				\hbox to 5.78in { {\bf Hybrid Systems} \hfill #2 }
				\vspace{4mm}
				\hbox to 5.78in { {\Large \hfill #5  \hfill} }
				\vspace{2mm}
				\hbox to 5.78in { {\it #3 \hfill #4} }
			}
		}
	\end{center}
	\vspace*{4mm}
}

\newcommand{\ho}[5]{\handout{#1}{#2}{Guide:
		#3}{#4}{#5}}
\newcommand{\al}{\alpha}
\newcommand{\nd}{\wedge}
\newcommand{\defn}{\coloneqq}
\newcommand{\Z}{\mathbb Z}
\newcommand{\real}{\mathbb{R}}
\newcommand{\nat}{\mathbb{N}}
\newcommand{\ppcd}{\mathcal{H}}
\newcommand{\reduce}[1]{#1^{red}}
\newcommand{\loc}{Q}
\newcommand{\state}{\mathcal{X}}
\newcommand{\poly}[1]{Poly(#1)}
\newcommand{\inv}{Inv}
\newcommand{\flow}{Flow}
\newcommand{\guard}{Guard}
\newcommand{\eguard}{\mathcal{G}}
\newcommand{\edges}{Edges}
\newcommand{\dist}[1]{\text{Dist}(#1)}
\newcommand{\itstar}{\item[$\bigstar$]}
\newcommand{\vertx}{V}
\newcommand{\gredge}{E}
\newcommand{\ewt}{W}
\newcommand{\ec}{E_c}
\newcommand{\ep}{E_p}
\newcommand{\graph}{G}
\newcommand{\initnode}{I_0}
\newcommand{\initst}{s_{\textit{init}}}
\newcommand{\stablest}{s_0}
\newcommand{\pf}{P} 
\newcommand{\Path}{\sigma}
\newcommand{\infpath}{\sigma_{\infty}}
\newcommand{\len}[1]{\text{len}(#1)}
\newcommand{\indx}{I}
\newcommand{\union}{\cup}
\newcommand{\bigunion}{\bigcup}
\newcommand{\intersect}{\cap}
\newcommand{\bigintersect}{\bigcap}
\newcommand{\ball}[2]{B_{#1}(#2)}
\newcommand{\prob}{P}
\newcommand{\syntaxpr}{Pr}
\newcommand{\probpath}[2]{P_{#1}(#2)}
\newcommand{\tildeprob}{\tilde{Pr}}
\newcommand{\fpath}[1]{\text{Paths}_{\textit{fin}}(#1)}
\newcommand{\cyl}[1]{Cyl(#1)}
\newcommand{\apath}[1]{\text{Paths}(#1)}
\newcommand{\ipath}[1]{\text{Paths}(#1)}
\newcommand{\spath}{\mathcal{SP}}
\newcommand{\scyl}{\mathcal{SC}}
\newcommand{\edgeset}{\mathcal{E}}
\newcommand{\sigal}{\mathcal{F}}
\newcommand{\states}{S}
\newcommand{\ctran}{\rightarrow_c}
\newcommand{\ptran}{\rightarrow_p}
\newcommand{\tran}{R}
\newcommand{\stran}{S_\rightarrow}
\newcommand{\sstran}[1]{{#1}_{\rightarrow}}
\newcommand{\metric}{d}
\newcommand{\weight}{W}
\newcommand{\M}{\mathcal{M}}
\newcommand{\restrict}[2]{#1|#2}
\newcommand{\wdtmc}{\mathcal{M}_W}
\newcommand{\boundary}[1]{\partial(#1)}
\newcommand{\norm}[2]{\lvert\!\lvert #1 \rvert\!\rvert_{#2}}
\newcommand{\pathset}[3]{\Sigma_{#1}^{#2}#3}
\newcommand{\partition}{\mathcal{P}}
\newcommand{\size}[1]{\vert #1\vert}
\newcommand{\abtrct}[2]{#1/#2}
\newcommand{\face}{\mathbb{F}}
\newcommand{\cycle}{\mathcal{C}}
\newcommand{\wld}{S_\sigma}
\newcommand{\ldecomp}[1]{#1^{\mathcal{L}}}
\newcommand{\fdecomp}[1]{#1^d}
\newcommand{\pcdecomp}[1]{#1^{\mathcal{SP}\union\mathcal{SC}}}
\newcommand{\tildewt}{\tilde{W}}
\newcommand{\rhost}{\rho^*}
\newcommand{\pst}[1]{p^{st}(#1)}
\newcommand{\mstep}[1]{\overset{\mathrm{#1}}{\leadsto}}
\newcommand{\cov}{Cov}
\newcommand{\var}{Var}
\newcommand{\linineq}[1]{L_{#1}}
\newcommand{\coreq}[1]{{#1}_{\textit{eq}}}

\newcommand{\ubar}[1]{\text{\b{$#1$}}}
\newcommand{\ap}{\textit{AP}}
\newcommand{\until}[1]{U^{\leq {#1}}}
\newcommand{\infuntil}{U}
\newcommand{\true}{\top}
\newcommand{\false}{\bot}
\newcommand{\opAnd}{\wedge}
\newcommand{\Or}{\vee}
\newcommand{\some}[1]{\Diamond^{\leq {#1}}}
\newcommand{\infsome}{\Diamond}
\newcommand{\all}[1]{\Box^{\leq {#1}}}
\newcommand{\infall}{\Box}
\newcommand{\ineqs}{\bowtie}
\newcommand{\next}{\bigcirc}
\newcommand{\Label}{L}
\newcommand{\assign}{V}
\newcommand{\pathvars}{\Pi}
\newcommand{\vecpi}{\bar{\pi}}
\newcommand{\vecpath}{\bar{\Path}}
\newcommand{\finpathvars}{\Pi_{fin}}
\newcommand{\semantics}[2]{{\llbracket #1\rrbracket}_{#2}}
\newcommand{\shift}[2]{#1^{(#2)}}
\newcommand{\notmodels}{\nvDash}
\newcommand{\thresh}[1]{\textit{threshold}_{#1}}
\newcommand{\nullhyp}{H_0}
\newcommand{\althyp}{H_1}
\newcommand{\sample}[1]{\mathbf{#1}}
\newcommand{\statistic}{T}
\newcommand{\falsepos}{\alpha_{\textit{FP}}}
\newcommand{\falseneg}{\alpha_{\textit{FN}}}
\newcommand{\range}[1]{[#1]}
\newcommand{\complmnt}[1]{{#1}^c}
\newcommand{\pred}{\mathbb{P}}
\newcommand{\pdf}{f}
\newcommand{\cdf}{F}
\newcommand{\bayes}{\mathbb{B}}
\newcommand{\samsp}{\Omega}
\newcommand{\extset}[2]{{#1}^{+}_{#2}}
\newcommand{\deductset}[2]{{#1}^{-}_{#2}}
\newcommand{\nocol}{\textit{nocol}}
\newcommand{\smc}{\text{SMC}}
\newcommand{\dtmc}{\text{DTMC}}
\newcommand{\hyppctl}{\text{HyperPCTL*}}
\newcommand{\sprt}{\text{SPRT}}
\newcommand{\toolname}{\text{HyProVer}}
\newcommand{\foravoid}{\psi_{\textit{ca}}}
\newcommand{\forreach}{\psi_{\textit{goal}}}
\newcommand{\Beta}{\text{Beta}}
\newcommand{\bernou}{\mathcal{B}}
\newcommand{\uniform}{\textit{U}}
\newcommand{\dbeta}{\textit{Beta}}
\newcommand{\mle}{\text{MLE}}
\newcommand{\tmle}[1]{{#1}^\textit{MLE}}
\newcommand{\argmax}{\text{argmax}}
\newcommand{\argmin}{\text{argmin}}
\newcommand{\kdiv}{\vert\!\vert}

\newcommand{\sd}[1]{{\color{black}#1}}
\newcommand{\pp}[1]{{\color{green}#1}}

%% file: intro.tex
\section{Introduction}
Formal verification of software systems that interact with physical environments as in cyber-physical systems, has gained significant importance in recent years owing to the safety-criticality of these systems.
Probabilistic and stochastic models have been proven to be useful tools in modeling the uncertainty inherent in the interactions with the environment.
In particular, discrete-time Markov chains are a simple, yet widely applicable, class of probabilistic systems that capture uncertainties using transition probabilities.
Correctness specifications of these systems need to incorporate the probabilistic aspects and are often captured using probabilistic logics. 
For example, uncertain behaviour of a car driver has been modeled as a discrete-time Markov chain ($\dtmc$) \cite{sadigh2014data},  and several behavioral properties of the driver are encoded using the Probabilistic Computation Tree Logic (PCTL). Further, a robot performing random walk on a grid has been modeled using $\dtmc$ \cite{lal2020bayesian},  and properties like probabilistic goal reaching have been encoded using Continuous Stochastic Logic (CSL). 

Traditional logics focus on properties about single execution traces of a system. 
While several interesting properties can be captured using single-trace logics, they fall short in capturing interactive behaviors between multiple agents.
Hyperproperties are a new class of properties \cite{clarkson2014temporal,hsu2021bounded} that capture multi-trace behaviors. 
 For example, consider two robots walking on an $n\times n$ grid. The property that the two robots never collide, can only be expressed by a formula that refers to the paths of both the robots simultaneously \cite{wang2020hyperproperties}.
In addition, several security properties can be easily captured using hyper-logics such as information flow \cite{agrawal2016runtime,clarkson2010hyperproperties}, non-interference \cite{finkbeiner2017monitoring,hsu2021bounded} and observational determinism \cite{finkbeiner2016deciding,finkbeiner2017monitoring}.

In this paper, we focus on the problem of verifying discrete-time Markov chains with respect to probabilistic hyperproperties that allow us to specify constraints on joint probability of satisfaction of real-time behaviors by independent executions of a multi-agent system. 
We focus on discrete-time systems and properties, however, the ideas in the paper will be foundational toward verification of these systems and properties in continuous time.
We consider robot navigation problems as our case studies.
For example, consider a two robot navigation scenario on a grid, wherein we desire to ensure an upper bound on the collision probability.
Such properties can be easily specified using probabilistic hyperproperty logics such as $\hyppctl$, which is an   
expressive hyperproperty logic that allows nesting of both temporal and probabilistic operators (see \cite{wang2021statistical}). 
Thus, it has found its use in formal specification of probabilistic hyperproperties of cyber-physical and  robotic systems \cite{wang2021statistical}.
Several quantitative extensions of information security properties have been captured using $\hyppctl$ such as 
qualitative information flow \cite{kopf2007information}, probabilistic non-interference \cite{gray1992toward} and differential privacy \cite{dwork2014algorithmic}.

Probabilistic model-checkers such as PRISM \cite{kwiatkowska2011prism} and STORM \cite{dehnert2017storm} verify probabilistic systems with respect to probabilistic properties.
While probabilistic model-checking performs reasonably well for single trace properties, even in the non-probabilistic setting, model-checking hyperproperties is a challenging task \cite{finkbeiner2018model}. 
Hence, light-weight verification methods based on sampling, such as statistical model-checking \cite{wang2019statistical,wang2021statistical} have been explored.

Statistical Model-Checking ($\smc$) \cite{clarke2008statistical,grosu2005monte,merayo2009statistical,rabih2009statistical,sen2004statistical,sen2005statistical,younes2002probabilistic,basu2009approximate,cappart2017verification,henriques2012statistical,zuliani2013bayesian} is an alternative sampling based method to verify satisfiability of the specification.
More precisely, $\smc$ algorithms collect a sample of execution paths from the system and determine the satisfiability of the property based on a statistic computed using the satisfiability of these sample execution paths.
$\smc$ scales extremely well for complex systems, however, it does not provide exact inference, rather, incurs Type-I and Type-II errors which correspond to the probability that, the test concludes that the system violates the property when indeed it is satisfied, and the probability that, the test concludes that the system satisfies the property when actually it is violated, respectively. 
$\smc$ for probabilistic hyperproperties has been discussed in \cite{wang2019statistical,wang2021statistical}, where sequential probability ratio test ($\sprt$) has been used. 
These algorithms fall into the frequentist regime where the parameter is assumed to be fixed but unknown.
In this paper, we explore an alternative approach based on Bayesian hypothesis testing which incorporates prior information about the parameters in the form of their distributions.
We present a Bayesian $\smc$ algorithm based on Bayes' test that takes as input Type-I and Type-II error bounds and provides corresponding guarantees on the inference of the satisfiability.


$\smc$ of probabilistic hyperproperties specified in $\hyppctl$ is challenging due to two reasons: $\hyppctl$ in interpreted over multiple traces and the probabilistic operators are nested.
To address the multi-trace setting, we present a multi-dimensional hypothesis testing and incorporate that within a recursive algorithm to tackle the nesting in probabilistic operators.
Our broad approach to verifying a $\hyppctl$ formula $\varphi$ is a bottom up algorithm similar to classical CTL model-checking \cite{clarke2018handbook,baier2008principles}, wherein we verify the bottom level probabilistic operators and work our way upwards using the results from the verification of the subformulae.
Since, $\smc$ only returns the correct answers with certain confidence, this error needs to be factored into the $\smc$ procedures for top level formulae.
We tackle this by designing a modified Bayes' test that uses a statistic computed from these approximately correct answers.
While such a framework has been developed in the single trace setting~\cite{lal2020bayesian}, the main contribution of the paper is the careful incorporation of multi-dimensional Bayesian hypothesis testing into the recursive framework by establishing appropriate bounds on the confidence of lower level multi-dimensional SMC calls.


We have implemented our algorithm in a Python toolbox $\toolname$ and tested it on a series of robot navigation scenarios in the grid world setting (see Section \ref{grid}). 
Though the algorithm is written in a bottom-up fashion, we have implemented an efficient top-down approach, wherein lower level SMC calls are made only on a need-by basis, thereby avoiding an exhaustive computation of lower level SMC calls.
We also implemented the SPRT approach for comparison.
We observed that, for the same level of confidence, our approach takes fewer number of samples and less time to verify both non-nested and nested $\hyppctl$ formulae compared to $\sprt$ \cite{wang2021statistical} (none of the case studies discussed in \cite{wang2021statistical} involved a nested probabilistic formula). Also, our approach does not need the assumption that the true probability lies outside the \emph{indifference region} (for non-nested formula) and thus, is more general than $\sprt$. Thus, our approach is more practical for verification of general $\hyppctl$ formulae.
Bayesian approaches are often criticized for their need for prior information.
We considered different choices for prior by considering various parameters for the $\Beta$ distribution and observed that priors effect the verification time in a minor manner.
However, since $\sprt$ does not use prior information, we used uniform prior while comparing Bayesian approach to $\sprt$, so that any parameter value is equally probable.


To summarize, the main contributions of this paper are:
\begin{itemize}
    \item A first Bayesian statistical model-checking algorithm for $\hyppctl$ that combines multi-dimensional Bayesian hypothesis testing and a recursive framework for error propagation of nested probabilistic operators.
    \item A novel top-down implementation that makes lower level SMC calls on-need and consists of additional bookkeeping to avoid redundant and unnecessary computation.
    \item Experimental evaluations and comparisons with existing approaches that demonstrate the scalability and benefits of our approach.
\end{itemize}

\emph{Organization:} The rest of the paper is organized as follows. We discuss related work in Section \ref{related}. Some basic definitions and notations are covered in Section \ref{prelim}. We define the model checking problem formally in Section \ref{problem}. The robot navigation system grid world and two properties of it are described in Section \ref{grid}. In Section \ref{hyp}, we discuss general and multi-dimensional hypotheses testing, as well as, Bayes' test, approximate Bayes' test and sequential probability ratio test ($\sprt$). Section \ref{algo} explains the recursive Bayesian $\smc$ algorithm and compares its approach with that of $\sprt$ based $\smc$ \cite{wang2021statistical}. In Section \ref{exp}, we evaluate our algorithm on the case study discussed in Section \ref{grid}. Finally, we conclude in Section \ref{conc}.

%% file: related.tex
\section{Related Works}\label{related}
Broadly, two classes of techniques have been explored for verification of probabilistic properties on probabilistic system models. 
Probabilistic model checking (PMC) techniques based on numerical methods that compute exact probability of satisfaction of a specification given in logics such as LTL \cite{vardi1985automatic}, PCTL \cite{abate2010approximate,ciesinski2004probabilistic}, CSL \cite{aziz1996verifying}, Bounded Temporal Logic (BTL) \cite{baier2003model} and $\omega$-regular languages \cite{bustan2004verifying} on stochastic models such as discrete-time Markov chain ($\dtmc$) \cite{abate2010approximate,ciesinski2004probabilistic,bustan2004verifying,vardi1985automatic}, continuous time Markov chain (CTMC) \cite{aziz1996verifying,baier2003model}, Markov decision process (MDP) \cite{ciesinski2004probabilistic,hermanns2008probabilistic} and $\omega$-automaton \cite{vardi1985automatic} have been explored. However, these methods are generally model specific and involve solving a system of linear equations \cite{ciesinski2004probabilistic,baier2003model} or complex properties of algebraic and transcendental number theory \cite{aziz1996verifying}, which makes them computationally intensive.\par
On the other hand, statistical model checking ($\smc$) algorithms are based on random sampling of the probabilistic models and use statistical tests to provide (with certain confidence) inference on whether probability of a property lies within a certain range. Although $\smc$ can only provide answers with certain amount of errors, it is much less computation intensive and thus used in a large number of real life applications \cite{cappart2017verification,jha2009bayesian}. $\smc$ has been used for verification of properties specified in LTL \cite{clarke2008statistical,grosu2005monte}, Bounded LTL \cite{henriques2012statistical}, Probabilistic Bounded LTL \cite{zuliani2013bayesian}, PCTL \cite{basu2009approximate} and CSL \cite{lal2020bayesian,sen2005statistical,younes2002probabilistic} where $\dtmc$ \cite{basu2009approximate,lal2020bayesian,sen2005statistical}, CTMC \cite{sen2005statistical}, MDP \cite{bogdoll2011partial,henriques2012statistical}, semi-Markov process (SMP) \cite{sen2005statistical,younes2002probabilistic}, generalised semi-Markov process (GSMP) \cite{younes2002probabilistic} and various Stochastic Hybrid Systems (SHS) \cite{clarke2008statistical,grosu2005monte,sen2004statistical,zuliani2013bayesian,merayo2009statistical,larsen1991bisimulation} have been considered for the underlying probabilistic models. A brief discussion on state of the art $\smc$ techniques (e.g., sampling and testing methods) can be found from several survey papers \cite{agha2018survey,legay2010statistical}. For example, sampling techniques like Monte-Carlo \cite{grosu2005monte} and perfect simulation \cite{rabih2009statistical} have been used to sample various probabilistic models, and statistical tests like Bayesian \cite{lal2020bayesian,zuliani2013bayesian}, importance sampling \cite{barbot2012coupling}, acceptance sampling \cite{younes2002probabilistic} have been used to gather inference about their probabilistic properties. Detailed study on various statistical tests and their applications has been performed \cite{hadjis2015importance,wald1945sequential}. Tools like Apmc \cite{peyronnet2006apmc}, PRISM \cite{kwiatkowska2011prism}, STORM \cite{dehnert2017storm} have been developed that can automatically verify probabilistic properties on probabilistic models.\par
Encoding multi-trace properties (known as hyperproperties) using logics like LTL, PCTL, CSL is not straightforward as they mainly capture properties about individual traces of the underlying model. Logics like Bounded HyperLTL \cite{hsu2021bounded}, HyperLTL \cite{clarkson2014temporal}, HyperCTL* \cite{clarkson2014temporal} were defined for specification of hyperproperties. More recently, HyperPSTL \cite{wang2019statistical} and $\hyppctl$ \cite{wang2021statistical} have been defined to encode probabilistic hyperproperties. Model checking of probabilistic hyperproperties is a relatively new area of research; one recent work explores a statistical model-checking algorithm based on Clopper-Pearson interval method \cite{wang2019statistical} and sequential probability ratio test ($\sprt$) \cite{wang2021statistical}. In this paper, we explore a Bayesian approach for verification of probabilistic hyperproperties. To the best of our knowledge, this is the first statistical model-checking algorithm for probabilistic hyperproperties based on a Bayesian approach.

%% file: Prelim.tex
\section{Preliminaries}\label{prelim}
Let us define some basic terms and notations that are used in the paper.
    Given a sequence $\Path = s_1,s_2,\dots$, $\Path[i]$ denotes $i^{th}$ element of the sequence $\Path$, that is, $s_i$, and $\size{\Path}$ denotes length of the sequence $\Path$.
    For an $n$-tuple $\bar{\sample{X}} = (X_1,\dots,X_n)$, $\bar{\sample{X}}[i]$ denotes the $i^{th}$ element of $\bar{\sample{X}}$, that is, $X_i$. For an $n$-tuple $\bar{\sample{X}}$, infinite norm is denoted by $\norm{\bar{\sample{X}}}{\infty}$ where, $\norm{\bar{\sample{X}}}{\infty} = \max_{i=1}^n\bar{\sample{X}}[i]$.
    For any natural number $n$, $[n]$ denotes the set of natural numbers $\{1,\dots,n\}$.
    
    Recall that, given a continuous random variable $X$, the function $F:\real\rightarrow[0,1]$ defined as $\cdf(x)=\prob(X\leq x)$ is called the \emph{cumulative distribution (cdf)} of $X$ and the function $\pdf:\real\rightarrow\real$ defined as $\pdf(x) = \frac{d}{dx}\cdf(x)$ is called the \emph{probability density (pdf)} of $X$.
    
    For any $x\in\real^n$, the $\epsilon$-ball around $x$ with respect to the infinite norm is defined as $\ball{\epsilon}{x}:=\{y\in\real^n: \norm{x-y}{\infty}\leq \epsilon\}$. For a set $A\subseteq\real^n$, the boundary of $A$ (denoted $\boundary{A}$) is defined as the set of all $x\in\real^n$ such that $\ball{\epsilon}{x}\intersect A\neq\emptyset$ and $\ball{\epsilon}{x}\intersect \complmnt{A}\neq\emptyset$, for all $\epsilon>0$.

%% file: Problem.tex
\section{$\hyppctl$ Verification Problem}\label{problem}
In this section, we formalize the model checking problem of a system given as a Discrete-time Markov Chain with respect to correctness criterion specified in the logic $\hyppctl$ \cite{wang2021statistical} which specifies hyperproperties.

\subsection{Discrete-time Markov Chain}
A discrete-time Markov chain is a structure that consists of a set of states along with transitions that specify the probability of the next state of the system given the current state.
\begin{definition}[Discrete-time Markov Chain]
	A discrete-time Markov Chain ($\dtmc$) is a tuple $\M=(\states,\tran,\ap,\Label)$ 
	where,
	\begin{itemize}
		\item $\states$ is the finite set of states;
		\item $\tran:\states\times\states\rightarrow[0,1]$ is the \emph{transition probability function} such 
		that for any $s\in\states$, $\sum_{s^\prime\in \states}\tran(s,s^\prime)=1$;
		\item $\ap$ is the set of atomic propositions; and
		\item $\Label:\states\rightarrow 2^{\ap}$ is the labeling function that associates a set of atomic 
		proposition to each state.
	\end{itemize}
\end{definition}
A path (trace) of a $\dtmc$ $\M$ is a sequence of states 
$\Path=s_1,s_2,\dots$ such that for all $i<\size{\Path}$, $\tran(s_i,s_{i+1})>0$. $\ipath{\M}$ denotes the set of all infinite paths and $\fpath{\M}$ denotes the set of all finite paths of a $\dtmc$ $\M$.

\subsection{$\hyppctl$ Logic}
Next, we define the syntax and semantics of $\hyppctl$ which is a logic for specifying probabilistic hyperproperties and is interpreted on multiple traces from a $\dtmc$. 

\subsubsection{Syntax of $\hyppctl$}
Let us fix a set of atomic propositions $\ap$ and a (possibly infinite) set of path variables $\pathvars$.
$\hyppctl$ formulae over $\ap$ and $\pathvars$ are defined by the following grammar:
\begin{align}\label{syntax}
	&\phi := a^\pi\mid\neg\phi\mid \phi\opAnd\phi\mid \next\phi\mid \phi 
	\until{k}\phi\nonumber\\
	&
	\begin{aligned}
	\   \   \   \   \   \   \   \   \   \   \   
	\mid 
	\pred_D(\syntaxpr^{\vecpi}(\phi),\dots,\syntaxpr^{\vecpi}(\phi)),
	\end{aligned}
\end{align}
where
	$a\in\ap$ is an atomic proposition and
	$\pi\in\pathvars$ is a path variable.
	$\next$ and $\until{k}$ are the ‘next’ and ‘until’ operators respectively, where $k 
	\in\nat\union\{\infty\}$ is the time bound. For this work, we will assume $k\in\nat$, i.e., we will only consider `until' operators with finite time bound.
	$\syntaxpr^{\vecpi}$ is the probability operator for a tuple of path variables 
	$\vecpi 
	= (\pi_1, \dots , \pi_n)$, where $n \in \nat$ and $\pi_i\in\pathvars$ for all $i\in [n]$.
	$\pred_D(x_1,\dots,x_n)$ is an $n$-ary predicate function which is satisfied iff $(x_1,\dots,x_n)\in D\subseteq \real^n$.
	
	Additional logic operators are derived as usual: $\true \equiv a^\pi \opAnd \neg a^\pi$, 
$\phi\wedge 
\phi^\prime \equiv \neg(\neg\phi\Or\neg\phi^\prime)$, $\phi \Rightarrow \phi^\prime \equiv 
\neg\phi\Or\phi^\prime$, $\some{k}\phi\equiv \true\until{k}\phi$, and 
$\all{k}\phi\equiv\neg\some{k}\neg\phi$. 
We represent a $1$-tuple by its element, i.e., 
$\syntaxpr^{(\pi)}$ is
written as $\syntaxpr^\pi$.
	
	\begin{remark}
	In general, $\hyppctl$ formulae are defined by the grammar:
	\begin{align*}
	    &\phi := a^\pi\mid\phi^\pi\mid\neg\phi\mid \phi\opAnd\phi\mid \next\phi\mid \phi \until{k}\phi\mid \rho\ineqs\rho\\
	    &\rho := f(\rho,\dots,\rho)\mid \syntaxpr^{\vecpi}(\phi)\mid \syntaxpr^{\vecpi}(\rho),
	\end{align*}
	where $\ineqs\in\{>,<,\geq,\leq,=\}$.
	Note that, a formula of the form $\rho_1\ineqs\rho_2$ can be equivalently depicted as $\pred_D(\rho_1,\rho_2)$, where $D=\{(x_1,x_2)\mid x_1-x_2\ineqs 0\}$. Similarly, other $\rho$ formulae can also be expressed by the grammar in Equation \ref{syntax} defined above, i.e., the grammar in Equation \ref{syntax} is as expressive as the general $\hyppctl$ grammar.
	\end{remark}

\subsubsection{Semantics of $\hyppctl$}
A $\hyppctl$ formula over $\ap$ and $\Pi$ is interpreted on the pair $(\M,\assign)$, where $\M$ is a $\dtmc$ with propositions $\ap$, and 
$\assign$:$\pathvars\rightarrow\ipath{\M}$ is a mapping from
$\pathvars$ to $\ipath{\M}$. We say $\assign$ is a \emph{path assignment}.
    $\semantics{\phi}{\assign}$ denotes instantiation of the mapping $\assign$ on the formula $\phi$, i.e., each $\pi\in\pathvars$ that appears in $\phi$ is replaced by the path $\assign(\pi)\in \ipath{\M}$.
    Let $\vecpi$ and $\vecpath$ be sequences of path variables and paths of $\M$ respectively such that $\size{\vecpi}=\size{\vecpath}$. $\assign[\vecpi\rightarrow\vecpath]$ denotes revision of the mapping $\assign$ where $\vecpi[i]$ is remapped to $\vecpath[i]$ for all $i\in[\size{\vecpi}]$. 
	$\shift{\assign}{i}$ denotes the $i$-shift of $\assign$, i.e., $\shift{\assign}{i}(\pi)$ is 
	the 
	path $(\assign(\pi)[i],\assign(\pi)[i+1],\dots)$.
	$\ipath{\assign,\pi}$ denotes the set of all 
	paths $\Path\in\ipath{\M}$ such that $\Path[0] = \assign(\pi)[0]$ and $\ipath{\assign,\vecpi}$ denotes the set of all path sequences $\vecpath\in(\ipath{\M})^*$ such that $(\vecpath[i])[0]=\assign(\vecpi[i])[0]$ for all $i$.
The semantics of $\hyppctl$ is described as follows,
\begin{align*}
	&(\M,\assign)\models a^\pi\quad\text{iff}\quad a\in \Label(\assign(\pi)(0))\\
	&(\M,\assign)\models \neg\phi\quad\text{iff}\quad (\M,\assign)\notmodels\phi\\
	&(\M,\assign)\models \phi_1\opAnd\phi_2\quad\text{iff}\quad (\M,\assign)\models\phi_1\text{ 
		and }(\M,\assign)\models\phi_2\\
	&(\M,\assign)\models\bigcirc\phi\quad\text{iff}\quad(\M,\shift{\assign}{1})\models\phi\\
	&(\M,\assign)\models\phi_1\until{k}\phi_2\quad\text{iff}\quad\exists i\leq k \text{ such that 
	}\\
	& \ \   \   \   \    \  \   \   \   \   
	((\M,\shift{\assign}{i})\models\phi_2)\opAnd(\forall j<i, (\M,\shift{\assign}{j})\models\phi_1)\\
	&(\M,\assign)\models\pred_D
	(\syntaxpr^{\vecpi_1}(\phi_1),\dots,\syntaxpr^{\vecpi_n}(\phi_n))
	\quad\text{iff}\\
	& \ \   \   \   \    \  \
	\quad(\semantics{\syntaxpr^{\vecpi_1}(\phi_1)}{\assign},\dots,\semantics{\syntaxpr^{\vecpi_n}(\phi_n)}{\assign}) \in D\\
	&\semantics{\syntaxpr^{\vecpi}(\phi)}{\assign}= \prob\{\vecpath \in \ipath{\assign,\vecpi}
	\mid 
	(\M, \assign[\vecpi\rightarrow\vecpath]) \models \phi\}
\end{align*}
The expression $\prob\{\vecpath \in \ipath{\assign,\vecpi}
	\mid 
	(\M, \assign[\vecpi\rightarrow\vecpath]) \models \phi\}$ denotes the probability of 
satisfaction of $\phi$ on the set of path sequences $\{\vecpath: \size{\vecpath}=\size{\vecpi}\text{ and }(\vecpath[i])[0] 
=\assign(\vecpi[i])[0]\  \forall i\}$.

\sd{
\subsection{The Model Checking Problem}
Let $\M$ be a $\dtmc$ and $\phi$ be a $\hyppctl$ formula. Let $\pathvars_\phi$ denote the finite set of path variables used to define $\phi$ and $\assign:\pathvars\rightarrow\ipath{\M}$ be a path assignment. 
The model checking problem is to decide whether $\M,\assign\models\phi$ given a $\dtmc$ $\M$ and path assignment $\assign$.
Note that, since we encounter path variables that appear in $\phi$ only and all linear time operators are bounded, it is enough to provide a path assignment $\assign:\pathvars_\phi\rightarrow \fpath{\M}$, as satisfaction of a formula depends only on finite prefixes of assignments.
}


%% file: case.tex
\section{Case Study: Grid World}\label{grid}
\begin{figure}[]
	\centering
	\setlength\abovecaptionskip{-10pt}
	\setlength\belowcaptionskip{-10pt}
	\includegraphics[width=10cm]{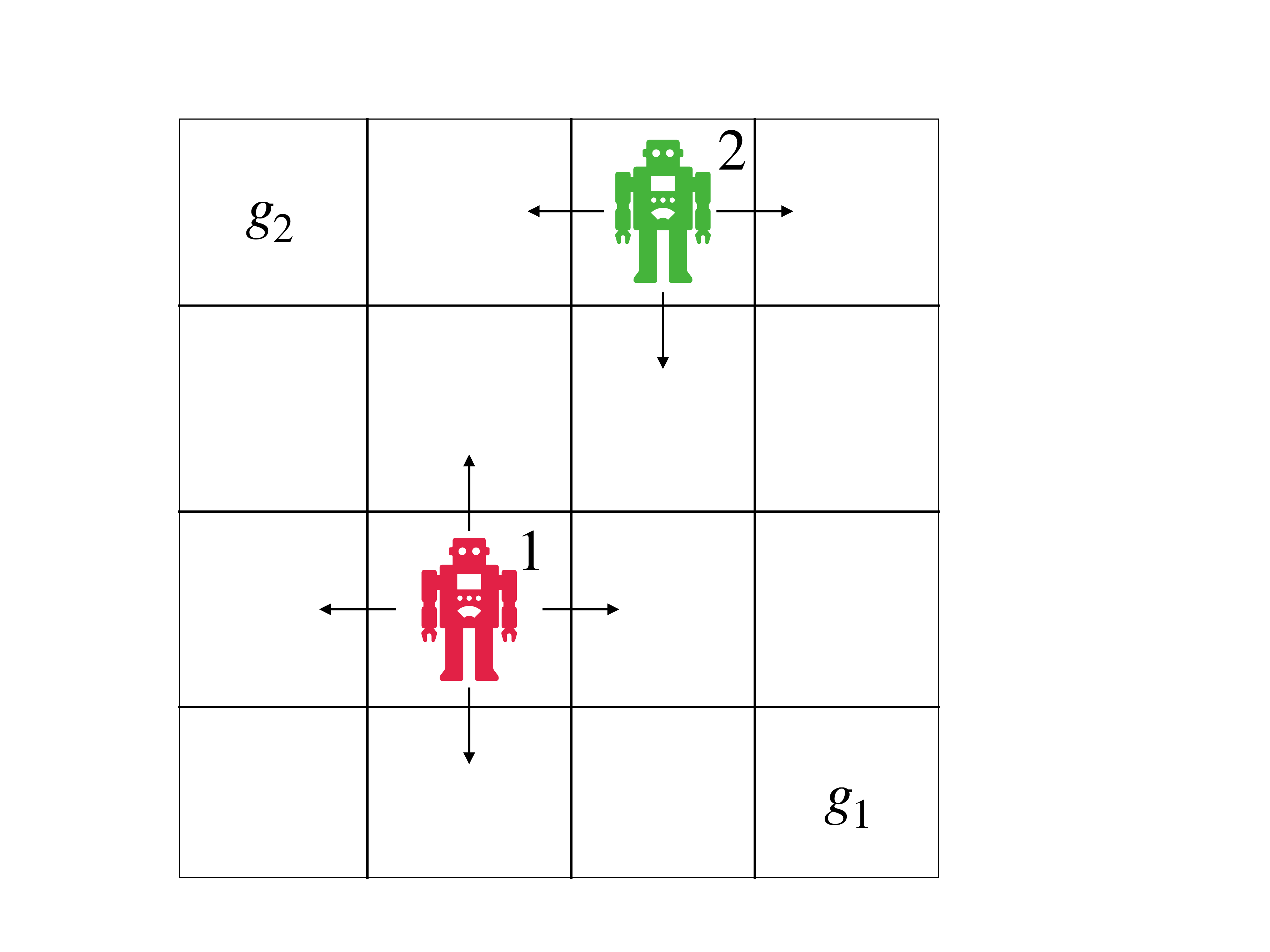}
	\caption{Robot Motion in Grid World}
	\label{grid_motion}
\end{figure}
In this section, we describe a robot navigation scenario on a grid world and discuss two desirable properties.
Consider an $n\times n$ grid where $N$ robots are performing $2$ dimensional random walks. 
In other words, each robot can move to a cell to its left, right, above or below only (if possible) with non zero probability. 
In Figure \ref{grid_motion}, a $4\times 4$ grid with two robots has been shown.
Each robot has a particular goal which is one of the cells of the grid. 
In Figure \ref{grid_motion}, the goal for the first robot (red) has been marked by $g_1$ and goal for the second robot (green) has been marked by $g_2$.
The first property we would like to verify is that the robots do not collide with each other within a finite number of steps with high probability.
The second property we would like to verify is that the robots reach their respective goals within a finite number of steps with high probability while avoiding collision with each other (with high probability).

\subsection{$\dtmc$ Modeling}
The robot navigation system can be aptly modeled using a discrete-time Markov chain ($\dtmc$), where, a state can uniquely represent the position and identity of a robot and the atomic propositions represent individual cells properties. In other words, a state is of the form $q_{ijk}$ where, $(i,j)$ denotes the position and $k$ denotes the identity of the robot. Thus, we have the set of states $\states = \{q_{ijk}\mid i,j\in[n] \text{ and }k\in [N]\}$. 
Each state $q_{ijk}$ is marked with the atomic proposition $a_{ij}$, which denotes the position corresponding to that state, i.e., $a_{ij}\in L(q_{ijk})$ for all $k$. 
Also, $q_{ijk}$ is marked with the proposition $g_k$ if $(i,j)$ is the goal of robot $k$, i.e., $g_k\in L(q_{ijk})$ iff $(i,j)$ is the goal of robot $k$. 
This gives us the set of all atomic propositions as $\ap  = \{a_{ij}\mid i,j\in [n]\}\union \{g_i\mid i\in [N]\}$ and the labeling function $\Label$. 
The transition function $\tran$ ensures that a robot can move from cell $(i,j)$ to cell $(i^\prime,j^\prime)$ only if they share a common boundary, i.e., 
$\tran(q_{ijk},q_{i^\prime j^\prime k^\prime})>0$ iff cells $(i,j)$ and $(i^\prime,j^\prime)$ share a common boundary and $k=k^\prime$. This probability might vary for different robots. 
Incorporating $k$ as an index of the states ensures that the probability $\tran(q_{ijk},q_{i^\prime j^\prime k})$ can be uniquely defined for each robot $k\in [N]$.

\subsection{Collision Avoidance}
A collision between two robots happens if they reach the same cell at the same point of time. Avoiding collision is a desirable property. In other words, assuming $\pi_1$ is the path followed by the first robot and $\pi_2$ is the path followed by the second robot on the $\dtmc$, we would like that $a_{ij}$ not in both $\pi_1[k]$ and $\pi_2[k]$ for all $k\in\nat$. Since we are only considering bounded formulae, we assume $k\leq K$ for some $K\in\nat$. Thus an interesting property of the grid world will be to check if the probability that $a_{ij}$ in both $\pi_1[k]$ and $\pi_2[k]$ for some $i,j\in [n]$ and for any $k\leq K$ is smaller than a certain threshold $\theta$. This can be represented by the $\hyppctl$ formula $\foravoid$:
\begin{equation}\label{col_avoid}
    \pred_{[0,\theta]}\left(\syntaxpr^{(\pi_1,\pi_2)} \left[\some{K} \left(\bigvee_{i,j\in[n]} (a_{ij}^{\pi_1}\land a_{ij}^{\pi_2})\right)\right]\right).
\end{equation}

\subsection{Collision Free Goal Reaching}
Another desirable property of the grid world is that a robot reaches its goal within some finite number of steps with high probability while avoiding collision with other robots with high probability. For an $n\times n$ grid with two robots, this property can be aptly represented for the first robot using the nested $\hyppctl$ formula $\forreach$:
\begin{align}\label{gol_reach}
    &\pred_{[\theta_1,1]}\left(\syntaxpr^{\pi_1} \left[(\psi_{\nocol})\until{K} (g_1^{\pi_1})\right]\right)\text{ where,}\\
    &\psi_{\nocol} = \pred_{[\theta_2,1]}\left(\syntaxpr^{\pi_2}\left[\neg\left(\bigvee_{i,j\in[n]} (a_{ij}^{\pi_1}\land a_{ij}^{\pi_2})\right)\right]\right)\nonumber.
\end{align}
Observe that, this formula represents the property that the first robot (following the path $\pi_1$) reaches its goal within $K$ steps with probability at least $\theta_1$ while avoiding collision with the second robot (following the path $\pi_2$) with probability at least $\theta_2$. We can state this property for the second robot as well in a similar manner by simply interchanging the positions of $\pi_1$ and $\pi_2$ and replacing $g_1$ by $g_2$.

%% file: Hyp.tex
\section{Hypothesis Testing}\label{hyp}
The objective of a hypothesis test is to make some inference on the parameter(s) of a probability distribution using some statistical tests. More precisely, consider a random variable $Y$ with pdf $\pdf$ that depends on some parameter $\Theta$. The goal of hypothesis testing is to determine whether $\Theta$ lies above or below a certain value $\theta_0$ known as the threshold. Thus we obtain two hypotheses:
\begin{equation*}
    \nullhyp:\Theta\geq\theta_0 \text{ vs } \althyp:\Theta<\theta_0.
\end{equation*}
To determine whether $\nullhyp$ or $\althyp$, referred to as the null and alternate hypothesis, respectively, is true, a hypothesis test consists of sampling the random variable $Y$ to obtain data $y$, which is a sequence of values of $Y$ observed while sampling, computing a statistic $\statistic(y)$, where $\statistic$ is a function that maps $y$ to some real number $\statistic(y)$, and accepting or rejecting $\nullhyp$ \sd{(with probability greater than a desired threshold)} based on the value of $\statistic(y)$.
The statistic $\statistic$ should be chosen such that the corresponding test provides the correct inference with high probability.
There are two types of errors associated with a hypothesis test: the probability of rejecting $\nullhyp$ when $\nullhyp$ is true, referred to as the Type-I error, and the probability of accepting $\nullhyp$ when $\nullhyp$ is false, referred to as the Type-II error. It is desirable that Type-I and Type-II errors are bounded above by some small positive quantities $\alpha, \beta \in(0,1)$ (respectively).

\subsection{Multidimensional Hypothesis Testing}
Our verification problem translates to a hypothesis testing problem on a vector of random variables (random vector) each with its own parameter.
Here, we set up the multi-dimensional hypothesis testing problem that involves a multi-dimensional parameter corresponding to a multi-dimensional random vector.
Let $\bar{\sample{Y}} = (Y_1,\dots,Y_n)$ be an $n$-dimensional random vector formed by independent random variables $Y_1, \cdots, Y_n$ with parameter vector $\bar{\sample{\Theta}}=(\Theta_1,\dots,\Theta_n)$, that is, $Y_i$ has parameter $\Theta_i$.  The goal of $n$-dimensional hypothesis testing is to decide 
whether $\bar{\sample{\Theta}}$ is in $D\subseteq \real^n$ or not. 
Hence, our hypothesis testing problem is to decide between the following two hypotheses:
\begin{equation}
\label{eq: orig_hyp}
	\nullhyp: \bar{\sample{\Theta}}\in D\quad\text{vs}\quad \althyp: \bar{\sample{\Theta}}\in\complmnt{D}.
\end{equation}
We obtain a sequence of $N$ samples $\bar{\sample{y}}=(\sample{y_1},\dots,\sample{y_N})$ from 
the $n$-dimensional random vector $\bar{\sample{Y}}=(Y_1,\dots,Y_n)$, that is, each $\sample{y_i}$ is a tuple of  
values $(y_{i1},\dots,y_{in})$ where $y_{ij}$ is sampled from the random variable $Y_j$. 
We compute a statistic $\statistic$ of $\bar{\sample{y}}$, and compare it with a value that depends on $D$ to 
determine whether $\nullhyp$ should be accepted or rejected.

\subsection{Bayesian Hypothesis Testing}
Bayesian hypothesis testing is a hypothesis testing framework wherein we assume some knowledge about the parameters in the form of prior distributions.
Bayesian methods often perform better in terms of the inference since they factor in additional information about the parameters.
Hence, we are given a random vector $\bar{\sample{Y}}$ of i.i.d. random variables along with the parameter vector $\bar{\sample{\Theta}}$ and a joint pdf $\pdf_{\bar{\sample{\Theta}}}$ that is known.
We compute a statistic known as the Bayes' factor, denoted  
$\bayes_{\bar{\sample{Y}}}(\bar{\sample{y}},D)$,  
given by the ratio of  the probability that the data $\bar{\sample{y}}$ is observed given $\nullhyp$ is true (denoted 
$\prob(\bar{\sample{y}}\mid 
\nullhyp)$) to the probability that 
$\bar{\sample{y}}$ is observed given $\althyp$ is true (denoted $\prob(\bar{\sample{y}}\mid\althyp)$). In other 
words,
\begin{align}\label{bayes_fact}
	\bayes_{\bar{\sample{Y}}}(\bar{\sample{y}},D) &= \frac{\prob(\bar{\sample{y}}\mid 
		\nullhyp)}{\prob(\bar{\sample{y}}\mid\althyp)}\nonumber\\
	&=\frac{\int_{\theta\in 
			D}\prob_{\bar{\sample{Y}}\mid\Theta}(\bar{\sample{y}}\mid\theta)\pdf_{\bar{\sample{\Theta}}}
		(\theta)d\theta}
	{\int_{\theta\in \complmnt{D}}
		\prob_{\bar{\sample{Y}}\mid\Theta}(\bar{\sample{y}}\mid\theta)\pdf_{\bar{\sample{\Theta}}}
		(\theta)d\theta}
	\cdot\frac{P_1}{P_0}
\end{align}
where $P_0 = \int_{\theta\in D}\pdf_{\bar{\sample{\Theta}}}(\theta)d\theta$ and $P_1 = \int_{\theta\in 
\complmnt{D}}\pdf_{\bar{\sample{\Theta}}}(\theta)d\theta$.


The following theorem \cite{lal2020bayesian,zuliani2013bayesian} provides us a way to perform hypothesis testing using Bayes' factor.


\begin{theorem}\label{bayes}
	[Bayes' test] Let $\bar{\sample{Y}}$ be a discrete random vector with parameter $\bar{\sample{\Theta}}$ 
	and corresponding joint pdf $\pdf_{\bar{\sample{\Theta}}}$. Let, \[\nullhyp:\bar{\sample{\Theta}}\in D \text{ vs }
	\althyp:\bar{\sample{\Theta}}\in\complmnt{D},\] be the null and alternative hypotheses respectively. Consider the test that
	\begin{itemize}
	    \item accepts $\nullhyp$ when 
	$\bayes_{\bar{\sample{Y}}}(\bar{\sample{y}},D)\geq\frac{1}{\beta}$ and
	\item rejects $\nullhyp$ when 
	$\bayes_{\bar{\sample{Y}}}(\bar{\sample{y}},D)\leq\alpha$.
	\end{itemize} 
	Then $\alpha$ and $\beta$ are the upper bounds of 
	Type-I and Type-II errors, respectively. 
\end{theorem}

In the sequel, the $Y_i$ in our hypothesis test will be a Bernoulli random variable with parameter $\Theta_i \in (0,1)$, denoted $Y_i \sim \bernou(\Theta_i)$.
We consider $\Beta$ distributions, denoted $\dbeta(a,b)$ with parameters $a, b$, for the prior, since they are defined on the interval $(0, 1)$ and can represent several important distributions for suitable choices of $a$ and $b$. For example, the uniform distribution on $[0,1]$, that is, $\uniform[0,1]$, can be represented by $\dbeta(1,1)$.

\subsection{Approximate Bayesian Hypothesis Testing}
We will encounter a situation while designing the Bayesian SMC where we cannot obtain samples from $Y_i$ itself, but can obtain samples from a distribution close to $Y_i$. 
Hence, we present a Bayesian hypothesis test for the parameters of $\bar{\sample{Y}}$ using samples from these approximate distributions.
Consider the case when sampling directly from the $n$-dimensional random vector $\bar{\sample{Y}}=(Y_1,\dots,Y_n)$ with parameter $\bar{\sample{\Theta}}$ is not possible. Suppose instead, we can sample from the $n$-dimensional random vector $\bar{\sample{Z}} = (Z_1,\dots,Z_n)$ with parameter vector $\bar{\sample{\Theta}}' = (\Theta'_1, \cdots, \Theta'_n)$ with the following property:
\begin{align}\label{corYZ}
	&\prob(Z_i=0\mid Y_i=1)\leq\alpha^\prime_i\quad\text{and}\nonumber\\ &\prob(Z_i=1\mid 
	Y_i=0)\leq\beta^\prime_i.
\end{align}
for all $i\in [n]$, where $\alpha^\prime_i,\beta^\prime_i\in(0,1)$ are small positive quantities. In other words, the distributions of $\bar{\sample{Y}}$ and $\bar{\sample{Z}}$ are ``close".

Given this relation between random variables $Y_i$ and $Z_i$ for each $i$, the following proposition (a direct corollary of Proposition $1$ in \cite{lal2020bayesian}) bounds the distance (associated with the infinite norm) between $\bar{\sample{\Theta}}$ and $\bar{\sample{\Theta}}^\prime$.
\begin{proposition}\label{cor_relate_theta}
    Let $\bar{\sample{Y}}$ and $\bar{\sample{Z}}$ be $n$-dimensional random vectors consisting of Bernoulli random variables, with parameters $\bar{\sample{\Theta}}$ and $\bar{\sample{\Theta}}^\prime$ respectively. In other words, $\bar{\sample{Y}}[i]\sim\bernou(\bar{\sample{\Theta}}[i])$ and $\bar{\sample{Z}}[i]\sim\bernou(\bar{\sample{\Theta}}^\prime[i])$ for all $i\in[n]$. Suppose Equation \ref{corYZ} holds for each $i\in[n]$, i.e.,
    \begin{align*}
	&\prob(\bar{\sample{Z}}[i]=0\mid \bar{\sample{Y}}[i]=1)\leq\alpha^\prime_i\quad\text{and}\\ &\prob(\bar{\sample{Z}}[i]=1\mid 
	\bar{\sample{Y}}[i]=0)\leq\beta^\prime_i, \quad\forall i\in[n].
    \end{align*}
    Then, $\norm{\bar{\sample{\Theta}}-\bar{\sample{\Theta}}^\prime}{\infty} \leq \delta$ where, $\delta = \max_{i=1}^n\{\max\{\alpha_i^\prime,\beta_i^\prime\}\}$.
\end{proposition}
\begin{proof}
	From Proposition 1 in \cite{lal2020bayesian} we have,
	\begin{align*}
	    &(1-\alpha^\prime_i)\bar{\sample{\Theta}}[i]\leq\bar{\sample{\Theta}}^\prime[i]\\
		\Rightarrow & (\bar{\sample{\Theta}}[i] - \bar{\sample{\Theta}}^\prime[i]) \leq \alpha^\prime_i\cdot\bar{\sample{\Theta}}[i] \leq \alpha^\prime_i,
	\end{align*}
	since $\Theta_i \leq 1$. Similarly, we also have,
	\begin{align*}
	&\bar{\sample{\Theta}}^\prime[i]\leq\bar{\sample{\Theta}}[i]+\beta^\prime_i(1-\bar{\sample{\Theta}}[i])\\
	\Rightarrow &(\bar{\sample{\Theta}}^\prime[i] - \bar{\sample{\Theta}}[i]) \leq \beta^\prime_i(1-\bar{\sample{\Theta}}[i]) \leq \beta^\prime_i,
	\end{align*}
	since $(1-\bar{\sample{\Theta}}[i])\leq 1$. Hence,
	\begin{align*}
	\size{\bar{\sample{\Theta}}[i] - \bar{\sample{\Theta}}^\prime[i]}&\leq\max\{(\bar{\sample{\Theta}}^\prime[i] - 
	\bar{\sample{\Theta}}[i]),(\bar{\sample{\Theta}}[i] - \bar{\sample{\Theta}}^\prime[i])\}\\
	&\leq\max\{\alpha^\prime_i,\beta^\prime_i\}
	\end{align*}
	for all $i\in[n]$. Thus, $\norm{\bar{\sample{\Theta}}-\bar{\sample{\Theta}}^\prime}{\infty} \leq \delta$ where $\delta =
    \max_{i=1}^n\{\max\{\alpha_i^\prime,\beta_i^\prime\}\}$.
\end{proof}

Our next goal is to device a test to deduce whether $\nullhyp:\bar{\sample{\Theta}}\in D$ or 
$\althyp:\bar{\sample{\Theta}}\in\complmnt{D}$ is satisfied using samples from 
$\bar{\sample{Z}}=(Z_1,\dots,Z_n)$.
To this end, let us first define $\epsilon$-expansion and $\epsilon$-reduction of a set $A$ as, $\extset{A}{\epsilon} = \{x\in\samsp\mid 
\ball{\epsilon}{x}\intersect A\neq \emptyset\}$ and $\deductset{A}{\epsilon} = 
\samsp\setminus\extset{(\complmnt{A})}{\epsilon}$ respectively, where $\samsp$ is the sample space for $\bar{\sample{\Theta}}$. In other words, $\extset{A}{\epsilon}$ is the set of all $x\in\samsp$ such that, the $\epsilon$-ball around $x$ has non-null intersection with the set $A$. Similarly, $\deductset{A}{\epsilon}$ is the complement set of all $x\in\samsp$ such that, the $\epsilon$-ball around $x$ has non-null intersection with complement of $A$. 
Provided $\prob(D)\in(0,1)$, we have the following theorem which provides us a test for $\bar{\sample{\Theta}}$, parameter of $\bar{\sample{Y}}$, using samples from $\bar{\sample{Z}}$.

\begin{theorem}\label{approx_bayes}[Approximate Bayes' test]
	Let $\bar{\sample{Y}}$ and $\bar{\sample{Z}}$ be discrete random vectors \sd{($\bar{\sample{Y}}[i],\bar{\sample{Z}}[i]$ are Bernoulli random variables for all $i$)} of same dimension with 
	parameters $\bar{\sample{\Theta}}$ and $\bar{\sample{\Theta}}^\prime$ respectively. 
	Also, let $\norm{\bar{\sample{\Theta}}-\bar{\sample{\Theta}}^\prime}{\infty}\leq \delta$. 
	Let, \[\nullhyp:\bar{\sample{\Theta}}\in D \text{ vs } \althyp:\bar{\sample{\Theta}}\in\complmnt{D}\] be the null and 
	alternate hypotheses respectively. 
	Consider the test that 
	\begin{itemize}
	\item accepts $\nullhyp$ when 
	$\bayes_{\bar{\sample{Z}}}(\bar{\sample{z}},\deductset{D}{\delta})\geq\frac{1}{\beta\cdot r_2}$ and 
	\item rejects $\nullhyp$ when 
	$\bayes_{\bar{\sample{Z}}}(\bar{\sample{z}},\extset{D}{\delta})\leq\alpha\cdot r_1$,
	\end{itemize}
	\sd{where $r_1$ and $r_2$ are constants defined as,
	\begin{equation*}
		r_1=\frac{\prob(\bar{\sample{\Theta}}\in D)}
		{\prob(\bar{\sample{\Theta}}\in \extset{D}{2\delta})} \quad\text{and}\quad
		r_2=\frac{\prob(\bar{\sample{\Theta}}\in \complmnt{D})}
		{\prob(\bar{\sample{\Theta}}\in \extset{(\complmnt{D})}{2\delta})}
	\end{equation*}
	Then, $\alpha$ and $\beta$ are the upper bounds of Type-I and Type-II errors respectively.}
\end{theorem}

\begin{proof}
Let us show that $\alpha$ is the Type-I error bound for the hypothesis test, that is, we show that $\prob(\{\text{reject }\nullhyp\}\mid \nullhyp) \leq \alpha$.

\sd{Note that, from Proposition \ref{cor_relate_theta}, $\bar{\sample{\Theta}}(\omega)\in D$ implies $\bar{\sample{\Theta}}^\prime(\omega)\in \extset{D}{\delta}$ and,  $\bar{\sample{\Theta}}^\prime(\omega)\in \extset{D}{\delta}$ implies 
$\bar{\sample{\Theta}}(\omega)\in \extset{D}{2\delta}$, 
as $\norm{\bar{\sample{\Theta}}-\bar{\sample{\Theta}}^\prime}{\infty}\leq\delta$. Hence, $\{\omega\mid \bar{\sample{\Theta}}(\omega)\in D\}\subseteq \{\omega\mid \bar{\sample{\Theta}}^\prime(\omega)\in \extset{D}{\delta}\}
\subseteq \{\omega\mid \bar{\sample{\Theta}}(\omega)\in\extset{D}{2\delta}\}$.}
\begin{align*}
    &\prob(\{\text{reject }\nullhyp\}\mid \nullhyp) = 
    \prob\left(\bayes_{\bar{\sample{Z}}}\left(\bar{\sample{z}},
    \extset{D}{\delta}\right)\leq\alpha\cdot r_1\mid \bar{\sample{\Theta}}\in D\right)\\
    &=\frac{\prob\left(\bayes_{\bar{\sample{Z}}}\left(\bar{\sample{z}},
    \extset{D}{\delta}\right)\leq\alpha\cdot r_1,\bar{\sample{\Theta}}\in D\right)}
    {\prob(\bar{\sample{\Theta}}\in D)}\\
    &\leq \frac{\prob\left(\bayes_{\bar{\sample{Z}}}\left(\bar{\sample{z}},
    \extset{D}{\delta}\right)\leq\alpha\cdot r_1,\bar{\sample{\Theta}}^\prime\in \extset{D}{\delta}\right)}
    {\prob(\bar{\sample{\Theta}}\in D)}\\
    &= \frac{\prob\left(\bayes_{\bar{\sample{Z}}}\left(\bar{\sample{z}},
    \extset{D}{\delta}\right)\leq\alpha\cdot r_1,\bar{\sample{\Theta}}^\prime\in \extset{D}{\delta}\right)}
    {\prob\left(\bar{\sample{\Theta}}^\prime\in \extset{D}{\delta}\right)}\cdot
    \frac{\prob(\bar{\sample{\Theta}}^\prime\in \extset{D}{\delta})}
    {\prob\left(\bar{\sample{\Theta}}\in D\right)}\\
    &\leq\sd{\prob\left(\bayes_{\bar{\sample{Z}}}\left(\bar{\sample{z}},
    \extset{D}{\delta}\right)\leq\alpha\cdot r_1\mid \bar{\sample{\Theta}}^\prime\in \extset{D}{\delta}\right) 
    \frac{\prob(\bar{\sample{\Theta}}\in \extset{D}{2\delta})}
    {\prob\left(\bar{\sample{\Theta}}\in D\right)}}\\
    &=\frac{\prob\left(\bayes_{\bar{\sample{Z}}}\left(\bar{\sample{z}},
    \extset{D}{\delta}\right)\leq\alpha\cdot r_1\mid \bar{\sample{\Theta}}^\prime\in \extset{D}{\delta}\right)}{r_1}\\
    &\leq \frac{\alpha\cdot r_1}{r_1}
    \cdot\prob\left(\bar{\sample{z}}\mid \bar{\sample{\Theta}}^\prime\in 
			\complmnt{(\extset{D}{\delta})}\right)\\
    &
    \begin{aligned}
    \   \   \   \   \   \   \   \   \   
    \text{[since }
    \bayes_{\bar{\sample{Z}}}(\bar{\sample{z}},\extset{D}{\delta})
    \leq r_1\alpha\text{ iff }
    \prob\left(\bar{\sample{z}}\mid \bar{\sample{\Theta}}^\prime\in 
			\extset{D}{\delta}\right)\\
			\leq r_1\alpha\cdot\prob\left(\bar{\sample{z}}\mid \bar{\sample{\Theta}}^\prime\in 
			\complmnt{(\extset{D}{\delta})}\right)
    \text{]}
    \end{aligned}\\
    &\leq \alpha\quad [\text{since }\prob\left(\bar{\sample{z}}\mid \bar{\sample{\Theta}}^\prime\in 
			\complmnt{(\extset{D}{\delta})}\right)\leq 1]
\end{align*}
\sd{Similarly, we can show that $\beta$ is the Type-II error bound for the hypothesis test, i.e., $\prob(\{\text{accept }\nullhyp\}\mid \althyp)\leq \beta$.} 

\end{proof}

\begin{remark}
Note that, Theorem \ref{approx_bayes} (approximate Bayes' test) essentially tests \[\nullhyp:\bar{\sample{\Theta}}'\in \deductset{D}{\delta} \text{ vs } \althyp:\bar{\sample{\Theta}}'\in (\extset{D}{\delta})^c.\] 
\sd{Now since $\prob(D)\in(0,1)$ and $D\subseteq \extset{D}{2\delta}$, $r_1$ is well defined and lies within $(0,1]$. Similarly, since $\prob(\complmnt{D}) = 1 - \prob(D)\in(0,1)$ and $\complmnt{D}\subseteq\extset{(\complmnt{D})}{2\delta}$, $r_2$ is well defined and also lies within $(0,1]$. 
This implies $\alpha r_1\leq \alpha$ and $\beta r_2\leq\beta$.
Thus, the acceptance/rejection conditions 
	$\bayes_{\bar{\sample{Z}}}(\bar{\sample{z}},\deductset{D}{\delta})\geq\frac{1}{\beta\cdot r_2}$ and 
	$\bayes_{\bar{\sample{Z}}}(\bar{\sample{z}},\extset{D}{\delta})\leq\alpha\cdot r_1$, using samples from $\bar{\sample{Z}}$, are stricter than those using $\bar{\sample{Y}}$. For acceptance, the region $D$ is shrunk and Bayes' statistic is expected to be larger than a larger threshold ($1/\beta r_2\geq 1/\beta$) and for rejection, the region $D$ is expanded and Bayes' statistic is expected to be smaller than a smaller threshold ($\alpha r_1\leq\alpha$).}
	
Further, the region $\extset{D}{\delta}\setminus\deductset{D}{\delta}$ is an \emph{indifference region} in the sense that, if $\bar{\sample{\Theta}}'\in\extset{D}{\delta}\setminus\deductset{D}{\delta}$, then Theorem \ref{approx_bayes} can neither accept nor reject $\nullhyp$. This is why, some nested formulae cannot be verified by the approximate Bayes' test. Hence, we should only use the approximate Bayes' test if $\bar{\sample{\Theta}}'$ is not too close to the boundary of the region $D$. Otherwise, approximate Bayes' test will not pass the acceptance/rejection criterion.  This problem arises when we verify nested $\hyppctl$ formulae using approximate Bayes' test, but not when we verify  non-nested formulae using the classical Bayes' test.

\end{remark}

\sd{
\subsection{Hypothesis Testing by SPRT}\label{sec: SPRT_hyp_test}
Since we are comparing our approach to $\sprt$ based $\smc$, we provide a short description of $\sprt$ based hypothesis testing \cite{wang2021statistical}.
In $\sprt$, parameter $\bar{\sample{\Theta}}$ is assumed to be fixed and we decide between two most \emph{indistinguishable hypotheses} instead of the original hypotheses (Equation \ref{eq: orig_hyp}). More precisely, we test 
\begin{equation*}\label{eq: indisting_hyp}
	\nullhyp': \bar{\sample{\Theta}}\in \deductset{D}{\epsilon}\quad\text{vs}\quad 
	\althyp': \bar{\sample{\Theta}}\in\complmnt{(\extset{D}{\epsilon})},
\end{equation*}
for some small $\epsilon>0$. 
A simpler statistic based on the log-likelihood function and Kullback-Leibler divergence is devised and $\nullhyp'$ or $\althyp'$ is accepted based on the position of the maximum likelihood estimate of $\bar{\sample{\Theta}}$ and the value of this statistic.
Note that, $\extset{D}{\epsilon}\setminus\deductset{D}{\epsilon}$ is an \emph{indifference region} in the sense that, if $\bar{\sample{\Theta}}\in \extset{D}{\epsilon}\setminus\deductset{D}{\epsilon}$, then we cannot test $\nullhyp$ vs $\althyp$ (Equation \ref{eq: orig_hyp}) using $\sprt$.

}

%% file: algo.tex
\section{Statistical Model Checking}\label{algo}
We will now discuss model checking of probabilistic hyperproperties using Bayes' test and approximate Bayes' test. Suppose a formula $\psi$, a $\dtmc$ $\M$ and a path assignment $\assign$ is given. We refer to a $\hyppctl$ formula as probabilistic if the top-level operator is $\pred_D$. Let us note that, only two cases might arise for a probabilistic $\hyppctl$ formula. On one hand, the formula can be non-nested, i.e., of the form $\pred_D(\syntaxpr^{\vecpi_1}(\phi_1),\dots,\syntaxpr^{\vecpi_n}(\phi_n))$ where no $\phi_i$ contains a probabilistic subformula. On the other hand, a formula can be nested, i.e., of the form $\pred_D(\syntaxpr^{\vecpi_1}(\phi_1),\dots,\syntaxpr^{\vecpi_n}(\phi_n))$ where some $\phi_i$ contains one or more probabilistic subformulae.

\subsection{Verifying Non-nested Probabilistic Formula}\label{sec: verify_non_nest}
As the base case, we will first discuss the verification of non-nested probabilistic formula on a $\dtmc$. Let us 
assume, we have a formula $\psi$ of the form 
$\pred_D(\syntaxpr^{\vecpi_1}(\phi_1),\dots,\syntaxpr^{\vecpi_n}(\phi_n))$ where each 
$\phi_i$ is non-probabilistic (without $\pred$ operator). We also assume each $\phi_i$ is closed by 
${\vecpi_i}$, i.e., $\pathvars_{\phi_i}\subseteq {\vecpi_i}$. Now let us define Bernoulli random variables 
$Y_{i=1,\dots,n}:\ipath{\assign,\vecpi_i}\rightarrow\{0,1\}$ as,
\begin{equation*}
	Y_i(\vecpath_i) =
	\begin{cases}
		1 \text{ if } (\M,\assign[\vecpi_i\rightarrow\vecpath_i])\models \phi_i\\
		0 \text{ otherwise}.
	\end{cases}
\end{equation*}
Then $\semantics{\syntaxpr^{\vecpi_i}(\phi_i)}{\assign} = \prob(Y_i = 1) = \Theta_i$, where 
$\Theta_i$ is the Bernoulli parameter of $Y_i$. Now our verification problem can be restated as a 
hypotheses testing problem,
\begin{equation*}
	\nullhyp: \bar{\sample{\Theta}}\in D\quad\text{vs}\quad \althyp: \bar{\sample{\Theta}}\in\complmnt{D},
\end{equation*}
where $\bar{\sample{\Theta}} = (\Theta_1,\dots,\Theta_n)$ is the parameter for the $n$ dimensional random vector $\bar{\sample{Y}} = (Y_1,\dots,Y_n)$. We say $\pdf$ is the prior for $\Theta_i$ if $\Theta_i$ is distributed with pdf $\pdf$. In Bayesian framework, $\pdf$ assumed to be known.
Assuming all $\Theta_i$ has the same prior $\pdf$, we can easily compute the prior for 
$\bar{\sample{\Theta}}$, say $\pdf_{\bar{\sample{\Theta}}}$, where 
$\pdf_{\bar{\sample{\Theta}}}(\theta_1,\dots,\theta_n) = \prod_{i=1}^n\pdf(\theta_i)$.
We can now apply Bayes' test repeatedly on larger and larger 
samples to deduce whether $\nullhyp$ or $\althyp$ holds with certain Type-I and Type-II error bounds. 
This procedure is described in the following algorithm.

\input{algo1}
The correctness of Algorithm \ref{non-nest} follows directly from Theorem \ref{bayes} (Bayes' test).
\begin{theorem}\label{correct_base}
[Correctness of BaseBayes]
    Let $\M$, $\psi$, $\assign$, $a$, $b$, $\alpha$ and $\beta$, be as in Algorithm \ref{non-nest}. 
    If $(\M,\assign)\models \psi$, then Algorithm
    \ref{non-nest} outputs True with probability at least $(1 -
    \alpha)$.
     If $(\M,\assign)\not\models\psi$, then Algorithm
    \ref{non-nest} outputs False with probability at least $(1 - \beta)$.
\end{theorem}

\subsection{Verifying Nested Probabilistic Formula}
We now describe our general $\smc$ algorithm BayesSMC$(\M,\psi,\assign,a,b,\alpha,\beta)$ which can verify nested probabilistic formula on a $\dtmc$. Here $\M$ is a $\dtmc$, $\psi$ is a (possibly) nested $\hyppctl$ formula, $\assign$ is a path assignment, $a,b$ are parameters for the $\Beta$ prior and $\alpha,\beta$ are allowed upper bounds for Type-I and Type-II errors. Let us start with an overview of the algorithm.

\subsubsection{Overview}
Consider a nested probabilistic formula $\psi$ of the form $\pred_D(\syntaxpr^{\vecpi_1}(\phi_1),\dots,\syntaxpr^{\vecpi_n}(\phi_n))$, where $\phi_i$ has probabilistic subformulae $\psi_{ij}$ for some $i\in[n]$. 
Given a $\dtmc$ $\M$ and a path assignment $\assign$, we want to verify if $(\M,\assign)\models\phi$ with Type-I and Type-II error bounds $\alpha,\beta$ respectively. 
We do this in a bottom up recursive approach where we first verify satisfaction of $\psi_{ij}$ using BayesSMC with Type-I and Type-II error bounds $\alpha_{ij},\beta_{ij}$ respectively. 
Then, we propagate the errors $\alpha_{ij},\beta_{ij}$ using error propagation rules (Property \ref{error_propagate}) in order to calculate $\alpha_i,\beta_i$, the Type-I and Type-II errors incurred by $\phi_i$.
Next, we calculate $\delta=\max_{i=1}^n\{\max\{\alpha_i,\beta_i\}\}$ from the errors incurred by the $\phi_i$ formulae.
Finally, we apply the approximate Bayes' test (Theorem \ref{approx_bayes}) using $\alpha$, $\beta$ and the derived $\delta$ to verify the satisfaction of $\psi$ on $(\M,\assign)$.

Let us explain this using the nested formula $\forreach$ (Equation \ref{gol_reach}). 
The formula tree for $\forreach$ depicting its subformulae are shown in Figure \ref{nest_tree}.
\begin{figure}[]
	\centering
	\setlength\abovecaptionskip{-10pt}
	\setlength\belowcaptionskip{-10pt}
	\includegraphics[width=10cm]{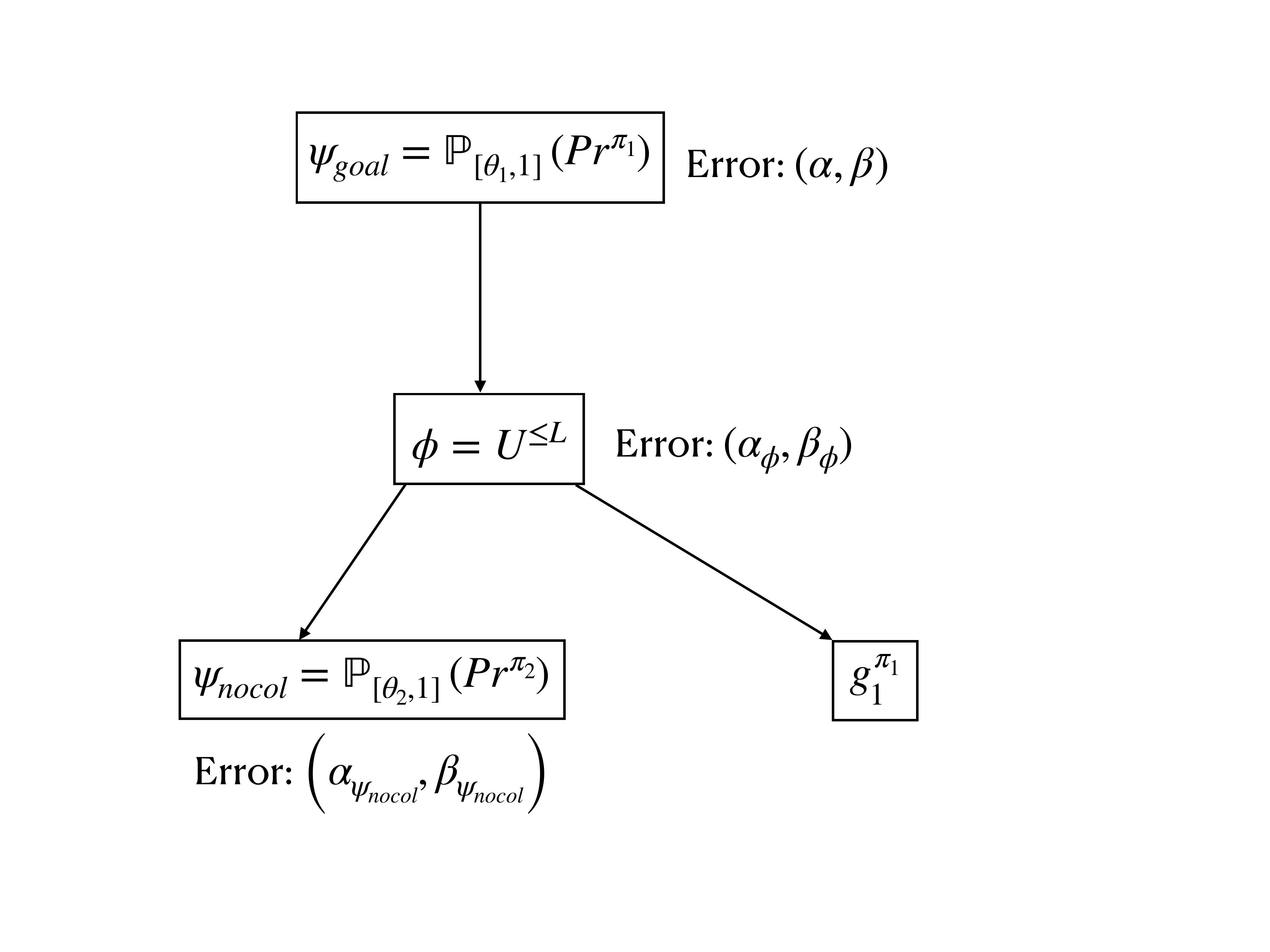}
	\caption{A nested probabilistic formula}
	\label{nest_tree}
\end{figure}
Suppose, we want to verify $\forreach$ on a $\dtmc$ $\M$ and a path assignment $\assign$ with error bounds $\alpha,\beta$. 
We first verify the non-nested formula $\psi_\nocol$ with error bounds $\alpha_{\psi_\nocol},\beta_{\psi_\nocol}$ on $(\M,\assign)$ using Bayes' test (Theorem \ref{bayes}).
Then, we calculate $\alpha_\phi,\beta_\phi$, Type-I and Type-II error bounds for the subformula $\phi = (\psi_{\nocol})\until{K} (g_1^{\pi_1})$, using error propagation rules (Property \ref{error_propagate}). Clearly, $\delta=\max\{\alpha_\phi,\beta_\phi\}$ as $\forreach$ has only one subformula $\phi$. Finally, we verify $\forreach$ on $(\M,\assign)$ using approximate Bayes' test (Theorem \ref{approx_bayes}), with parameters $\alpha$, $\beta$ and $\delta$.

\subsubsection{Error Propagation}
\label{sec:error}
\sd{Let us define the recursive rules for error propagation now. 
Error is propagated in a bottom up manner, that is, from a subformula to its parent formula.
For a formula $\psi$, let $E_1(\psi)$ denotes the Type-I error associated to $\psi$ and $E_2(\psi)$ denotes the Type-II error associated to $\psi$. We now describe error propagation for a general (possibly nested) $\hyppctl$ formula.}

\sd{
Let $\psi = \pred_D(\syntaxpr^{\vecpi_1}(\phi_1),\dots,\syntaxpr^{\vecpi_n}(\phi_n))$ be a (possibly nested) $\hyppctl$ formula, where each $\phi_i$ can have zero or more probabilistic subformula. Let us assume, by inductive hypothesis, $E_1(\phi_i)$ and $E_2(\phi_i)$ are Type-I and Type-II errors associated to $\phi_i$. Let $\delta=\max_{i=1}^n\{\max\{E_1(\phi_i),E_2(\phi_i)\}\}$. Any $E_1(\phi_i)$ and $E_2(\phi_i)$ in $(0,1)$ are allowed as long as $\deductset{D}{\delta}\neq\emptyset$ (this is a necessary condition because to apply approximate Bayes' test, one needs to compute $\bayes_{\bar{\sample{Z}}}(\bar{\sample{z}},\deductset{D}{\delta})$ for some random vector $\bar{\sample{Z}}$, and $\bayes_{\bar{\sample{Z}}}(\bar{\sample{z}},\deductset{D}{\delta})$ is undefined in case $\deductset{D}{\delta}=\emptyset$). Thus, we have a complete recursive definition of error propagation for a general $\hyppctl$ formula given by the following property,
\begin{property}\label{error_propagate}
\begin{enumerate}
    \item $E_1(a^\pi)=E_2(a^\pi)=0$ for all $a\in\ap$ and $\pi\in\pathvars$;
    \item $E_1(\neg\psi) = E_2(\psi)$, $E_2(\neg\psi) = E_1(\psi)$;
    \item $E_1(\bigcirc\psi) = E_1(\psi)$, $E_2(\bigcirc\psi) = E_2(\psi)$;
    \item $E_1(\psi_1 \land \psi_2) = E_1(\psi_1) + E_1(\psi_2)$, $E_2(\psi_1 \land \psi_2) = \max\{E_2(\psi_1), E_2(\psi_2)\}$;
    \item $E_1(\psi_1\until{k}\psi_2) = k\cdot E_1(\psi_1) + E_1(\psi_2)$, $E_2(\psi_1\until{k}\psi_2) = (k + 1) \max\{E_2(\psi_1), E_2(\psi_2)\}$.
    \item When $\psi$ is nested, $E_1(\psi),E_2(\psi)$ can be any value in $(0,1)$, whereas, $\delta$ is the maximum of all $\alpha_i,\beta_i$, where $\alpha_i,\beta_i$ are error bounds for the subformulae $\phi_i$.
\end{enumerate}
\end{property}
Note that, rules $1$-$5$ describe recursive error propagation for temporal formulae \cite{lal2020bayesian}, whereas, rule $6$ describes error propagation for (possibly nested) probabilistic formulae.

}

\subsubsection{The Recursive Algorithm BayesSMC}

Let $\psi=\pred_D(\syntaxpr^{\vecpi_1}(\phi_1),\dots,\syntaxpr^{\vecpi_n}(\phi_n))$ be a nested formula where, each $\phi_i$ can consist of zero or more probabilistic subformulae. We define Bernoulli random variables $Y_i$ as before. However, we cannot sample $Y_i$ directly as we use SMC to check whether $(\M,\assign[\vecpi_i\rightarrow\vecpath_i])\models \phi_i$ which introduces uncertainty (since $\phi_i$ itself contains zero or more probabilistic subformulae). 
Thus, we sample from Bernoulli random variables $Z_{i=1,\dots,n}:\ipath{\assign,\vecpi_{i}}\rightarrow\{0,1\}$ where,

\begin{align*}
	&Z_i(\vecpath) =
	\begin{cases}
		1 \text{ if } \text{BayesSMC}(\M,\phi_i,\assign^\prime,a,b,\alpha_i^\prime,\beta_i^\prime) = \text{True}\\
		0 \text{ otherwise},
	\end{cases}
\end{align*}
and $\assign^\prime = \assign[\vecpi_{i}\rightarrow\vecpath_{i}]$.
\sd{However,} in this process we incur some errors with non-zero probability. More precisely, the Type-I error is given by $\prob(Z_i=0\mid (\M,\assign^\prime)\not\models\phi_i)$ and Type-II error is given by $\prob(Z_i=1\mid (\M,\assign^\prime)\models\phi_i)$.
 \input{algo2} 
 Observe that, Type-I error is exactly equal to $\prob(Z_i=0\mid Y_i=1)$ and Type-II error is exactly equal to $\prob(Z_i=1\mid Y_i=0)$. Since by inductive hypothesis, Type-I and Type-II errors of BayesSMC$(\M,\phi_i,\assign^\prime,a,b,\alpha_i^\prime,\beta_i^\prime)$ are bounded by $\alpha_i^\prime$ and $\beta_i^\prime$ respectively, we have,
\begin{align*}
    &\prob(Z_i=0\mid Y_i=1)\leq\alpha^\prime_i\quad\text{and}\nonumber\\ &\prob(Z_i=1\mid 
	Y_i=0)\leq\beta^\prime_i.
\end{align*}
Thus we can say, the distributions of $\bar{\sample{Y}}=(Y_1,\dots,Y_n)$ and $\bar{\sample{Z}}=(Z_1,\dots,Z_n)$ are ``close" by Equation \ref{corYZ}. We can now apply Theorem \ref{approx_bayes} (approximate Bayes' test) to devise the recursive algorithm BayesSMC for verifying a (possibly) nested formula $\psi$ on a $\dtmc$ $\M$.

The correctness of Algorithm \ref{nest} follows directly from Theorem \ref{approx_bayes} (approximate Bayes' test).
\begin{theorem}\label{correct_smc}
[Correctness of BayesSMC]
    Let $\M$, $\psi$, $\assign$, $a$, $b$, $\alpha$ and $\beta$, be as in Algorithm \ref{nest}. 
    If $(\M,\assign)\models \psi$, then Algorithm
    \ref{nest} outputs True with probability at least $(1 -
    \alpha)$.
     If $(\M,\assign)\not\models\psi$, then Algorithm
    \ref{nest} outputs False with probability at least $(1 - \beta)$.
\end{theorem}

\sd{
\paragraph*{Comparison with SPRT based SMC}
Let us compare our approach with $\sprt$ based $\smc$ \cite{wang2021statistical}. 
In $\sprt$ based $\smc$, for non-nested probabilistic formula, the verification problem is mapped to an equivalent $n$-dimensional hypothesis testing problem and solved by $\sprt$ based hypothesis testing (Section \ref{sec: SPRT_hyp_test}); which is similar to our approach. On the other hand, for nested formula, $\sprt$ based $\smc$ uses verification results of subformulae directly while verifying the main formula. Thus, the total error incurred depends on the number of samples required to verify the main formula, which is not true for Bayesian $\smc$. Note that, both approaches provide Type-I and Type-II confidences for verification of probabilistic formulae.

}

%% file: algo1.tex
\begin{algorithm}[H]
	\begin{algorithmic}[1]
	\renewcommand{\algorithmicrequire}{\textbf{Input:}}
    \renewcommand{\algorithmicensure}{\textbf{Output:}}
		\Require $\M$: $\dtmc$, $\psi = \pred_D(\syntaxpr^{\vecpi_1}(\phi_1),\dots,\syntaxpr^{\vecpi_n}(\phi_n))$: non-nested $\hyppctl$ formula, $\assign$: path assignment, $a,b$: parameters for $\Beta$ prior, $\alpha,\beta$: bounds on Type-I and Type-II errors
		\Ensure Answer if $\M,\assign\models\psi$ with confidence $\alpha,\beta$
		\State $N\gets 1$
		\While{True}
		    \State /*Generate data $\bar{\sample{y}}=(\sample{y_1},\dots,\sample{y_N})$ for $\bar{\sample{Y}}$*/
		    \For{$k = 1$ to $N$}
		       \For{$i=1$ to $n$} 
		        \State Randomly sample $\vecpath_i$ from $\ipath{\assign,\vecpi_i}$
		        \If{$(\M,\assign[\vecpi_i\rightarrow\vecpath_i])\models \phi_i$}
		            \State $\sample{y_k}[i]\gets 1$
		        \Else
		            \State $\sample{y_k}[i]\gets 0$
		        \EndIf
		       \EndFor
		    \EndFor
			\State Calculate $\bayes_{\bar{\sample{Y}}}(\bar{\sample{y}},D)$ using Equation \ref{bayes_fact}
			\If{$\bayes_{\bar{\sample{Y}}}(\bar{\sample{y}},D) \geq 1/\beta$}
			\State Return True
			\ElsIf{$\bayes_{\bar{\sample{Y}}}(\bar{\sample{y}},D)\leq \alpha$}
			\State Return False
			\Else
			\State $N\gets 2\cdot N$
			\EndIf
		\EndWhile
	\end{algorithmic}
	\caption{BaseBayes: $\smc$ of non-nested $\hyppctl$ formula}
	\label{non-nest}
\end{algorithm}

%% file: algo2.tex
\begin{algorithm}[H]
	\begin{algorithmic}[1]
	\renewcommand{\algorithmicrequire}{\textbf{Input:}}
    \renewcommand{\algorithmicensure}{\textbf{Output:}}
		\Require $\M$: $\dtmc$, $\psi = \pred_D(\syntaxpr^{\vecpi_1}(\phi_1),\dots,\syntaxpr^{\vecpi_n}(\phi_n))$: (possibly) nested $\hyppctl$ formula, $\assign$: path assignment, $a,b$: parameters for $\Beta$ prior, $\alpha,\beta$: bounds on Type-I and Type-II errors
		\Ensure Answer if $\M,\assign\models\psi$ with confidence $\alpha,\beta$
		\sd{
		\If{$\psi$ is non-probabilistic}
		    \State Compute $\alpha,\beta$ from subformulae (Property \ref{error_propagate})
		    \If{$(\M,\assign)\models\psi$}
		        \State Return True
		    \Else
		        \State Return False
		    \EndIf
		\EndIf
		}
		\If{$\psi$ is non-nested}
		    \State Return BaseBayes$(\M,\psi,\assign,a,b,\alpha,\beta)$
		\EndIf
		\State $N\gets 1$
		\While{True}
		    \State /*Generate data $\bar{\sample{z}}=(\sample{z_1},\dots,\sample{z_N})$ for $\bar{\sample{Z}}$*/
		    \For{$i = 1$ to $n$}
		        \State \sd{/*Recursively compute error bounds for $\phi_i$ 
		        \Statex $\ \ \ \ \ \ \ \ $
		        by Property \ref{error_propagate}*/}
		        \State $\alpha_i^\prime \gets$ Type-I error bound of $\phi_i$
		        \State $\beta_i^\prime \gets$ Type-II error bound of $\phi_i$
		       \For{$k=1$ to $N$} 
		        \State Randomly sample $\vecpath_i$ from $\ipath{\assign,\vecpi_i}$
		        \State $\assign^\prime\gets \assign[\vecpi_i\rightarrow\vecpath_i]$
		        \If{BayesSMC$(\M,\phi_i,\assign^\prime,a,b,\alpha^\prime_i,\beta^\prime_i)$}
		            \State $\sample{z_k}[i]\gets 1$
		        \Else
		            \State $\sample{z_k}[i]\gets 0$
		        \EndIf
		       \EndFor
		    \EndFor
		    \State $\delta\gets \max_{i=1}^n\{ \max\{\alpha_i^\prime,\beta_i^\prime\}\}$
		    \State Calculate constants $r_1,r_2$ from Theorem \ref{approx_bayes}
			\State Calculate $\bayes_{\bar{\sample{z}}}(\bar{\sample{z}},\extset{D}{\delta})$ and $\bayes_{\bar{\sample{z}}}(\bar{\sample{z}},\deductset{D}{\delta})$ (by Equation \ref{bayes_fact})
			\If{$\bayes_{\bar{\sample{z}}}(\bar{\sample{z}},\deductset{D}{\delta}) \geq 1/({\beta\cdot r_2})$}
			\State Return True
			\ElsIf{$\bayes_{\bar{\sample{z}}}(\bar{\sample{z}},\extset{D}{\delta})\leq \alpha\cdot r_1$}
			\State Return False
			\ElsIf{$\bayes_{\bar{\sample{z}}}(\bar{\sample{z}},\extset{D}{\delta})\geq 1/({\beta\cdot r_2})$}
			\If{$\bayes_{\bar{\sample{z}}}(\bar{\sample{z}},\deductset{D}{\delta})\leq \alpha\cdot r_1$}
			    \State Return Undecided
			\EndIf
			\Else
			\State $N\gets 2\cdot N$
			\EndIf
		\EndWhile
	\end{algorithmic}
	\caption{BayesSMC: $\smc$ of (possibly) nested $\hyppctl$ formula}
	\label{nest}
\end{algorithm}

%% file: experiment.tex
\section{Experimental Evaluation}\label{exp}
We evaluated our approach on the grid world robot navigation system discussed in section \ref{grid}. 
We consider $n\times n$ grids with two robots, for varying grid sizes $n$.
The robots start from diagonally opposite cells, and their respective goals are to reach the horizontally opposite cells starting from their initial positions. 
We consider the collision avoidance property specified by the non-nested Formula $\foravoid$ (Equation \ref{col_avoid}) and the collision free goal reaching property for the first robot specified by the nested Formula $\forreach$ (Equation \ref{gol_reach}).

We have implemented our algorithm in the Python tool box $\toolname$.
The recursive algorithm proceeds in a bottom-up fashion, where we need to check satisfiability of the subformulae and 
work our way up. 
However, we do not a priori know all the assignments on which the subformulae need to be evaluated, hence, a bottom-up approach would be expensive if we were to compute the satisfiability of the subformulae on all possible assignments.
Instead, we have implemented an equivalent top-down algorithm where we start from the top and work our way down and evaluate the subformulae on only those assignments that are propagated down from the samples for the top-level formula.
For the $\pred$ operator, we only consider box constraints. Thus, integral can be evaluated on each dimension using the incomplete beta function and multiplied, to obtain the $n$-dimensional integrals required for calculating the Bayes' factor.
We compare our Bayesian approach with our own implementation of the Frequentist approach based on the $\sprt$ method \cite{wang2021statistical}.
Note that, there are no publicly available probabilistic model checkers for $\hyppctl$ for us to compare with.

Our verification results are summarized in Table \ref{col_avoid_table} and Table \ref{gol_reach_table} for Formula $\foravoid$ (Equation \ref{col_avoid}) and Formula $\forreach$ (Equation \ref{gol_reach}), respectively, wherein we report the verification time in seconds and number of samples required for deduction of the satisfiability of the topmost probabilistic formula. 
Note that, the reported time and number of samples are the average values over multiple (50) runs of the same experiment.
All experiments are performed on a machine having macOS Big Sur with Quad-Core Intel Core i7 2.8GHz$\times$ 1 Processor and 16GB RAM.
We run each formula for at most $30$ minutes and terminate the model-checker if it cannot provide a decision by that time.

Table \ref{diff_prior} compares the Bayesian approach for different priors which are obtained by instantiating the $\Beta$ distribution parameters $a$ and $b$. We used the uniform ($a=b=1$), a left-skewed ($a=5,b=2$), a right-skewed ($a=2,b=5$) and a bell-shaped ($a=b=2$) distribution as $\Beta$ priors. 
For $\alpha=\beta=0.01$, we verified the nested Formula $\forreach$ (Equation \ref{gol_reach}) describing collision free goal reaching for the first robot. We used different values of $n$, while keeping $K=8$, $\theta_1=0.5$ and $\theta_2=0.5$ fixed. 
Also for $n=4,8$, we moved the goal of the first robot to grid position $(1,1)$, so that the formula is satisfied.
We observe that, the choice of prior affects the verification time and number of samples required by the topmost formula (averaged over $50$ runs) in a minor manner. 
We use uniform prior for further comparison with $\sprt$, so that all parameter values are equally probable. 
Note that, a non-null indifference region always exists for $\sprt$ (measured by the parameter $\epsilon$) and like Bayesian, it also provides Type-I and Type-II guarantees for the correctness of a verification result \cite{wang2021statistical}.

\input{table3}

\input{table1}

The collision avoidance formula $\foravoid$ (Equation \ref{col_avoid}) depends on $3$ parameters: $n$ (grid size), $K$ (bound for the until operator) and $\theta$ (probability threshold). In Table \ref{col_avoid_table}, we varied $n$ and $K$ for different error bounds $\alpha, \beta$ and kept $\theta=0.5$ fixed, as threshold value has little bearing on the verification time and required number of samples (from our observations as well as existing work on Bayesian SMC for Continuous Stochastic Logic \cite{lal2020bayesian}). We would like to note two main observations from Table \ref{col_avoid_table}:
\begin{enumerate}
    \item For non-nested formula, $\sprt$ could not decide within time limit ($30$ minutes) whether collision probability lies below the threshold $\theta$ when collision probability was exactly $0$. This is because we can only separate $\nullhyp$ from $\althyp$ using $\sprt$ when $\bar{\sample{\Theta}}$ does not lie in the \emph{indifference region} (see \cite{wang2021statistical}) and that is not true here. If the test region is $D=[0,\theta]$, then the indifference region, however small it might be, will always contain $0$. Bayesian method does not have this problem, and it was able to provide inference in all the cases where $\sprt$ failed. This shows a benefit of the Bayesian approach.
    
    \item In those cases where both methods provided inference, Bayesian approach examined fewer samples and terminated in shorter time in most of the cases as compared to the $\sprt$ approach. 
    This shows Bayesian approach is much more scalable than $\sprt$ even for non-nested formulae. Note that, the inference provided by the two approaches always agree.
\end{enumerate}

\input{table2}

The collision free goal reaching formula for the first robot, $\forreach$ (Equation \ref{gol_reach}), depends on $4$ parameters: $n$ (grid size), $K$ (bound for until operator), $\theta_1$ (probability threshold for the topmost probabilistic formula) and $\theta_2$ (probability threshold for the probabilistic subformula $\psi_{nocol}$). In Table \ref{gol_reach_table}, we varied $n$ and $K$ for different error bounds $\alpha,\beta$ and kept $\theta_1=0.3$ and $\theta_2=0.5$ fixed for the same reasons as mentioned before. 
Also for $n=4,8$, we moved the goal of the first robot to grid position $(0,1)$, so that the formula is satisfied.
We see that, the performance of $\sprt$ is much worse for nested probabilistic formula as the verification never completes within the stipulated time limit (30 minutes) for any $n$, $K$ and $(\alpha,\beta)$. 
The derogatory performance for the nested case is expected, since, the verification time and the corresponding number of samples grow exponentially with the nesting depth.
Note that, $\sprt$ was not tested on any nested $\hyppctl$ formula in \cite{wang2021statistical} as well. Thus, our evaluation demonstrates that Bayesian approach is superior to the Frequentist SMC for verification of hyperproperties and scales reasonably well for nested formulae as well. 

%% file: table3.tex
\begin{table*}[]
\resizebox{\textwidth}{!}{%
\begin{tabular}{|c|cc|cc|cc|cc|c|}
\hline
\multicolumn{1}{|l|}{} &
  \multicolumn{2}{c|}{Uniform prior} &
  \multicolumn{2}{c|}{Left-skewed prior} &
  \multicolumn{2}{c|}{Right-skewed prior} &
  \multicolumn{2}{c|}{Bell-shaped prior} &
  \multicolumn{1}{l|}{} \\ \hline
n &
  \multicolumn{1}{c|}{Samples} &
  Time &
  \multicolumn{1}{c|}{Samples} &
  Time &
  \multicolumn{1}{c|}{Samples} &
  Time &
  \multicolumn{1}{c|}{Samples} &
  Time &
  Status \\ \hline\hline
$4$ &
  \multicolumn{1}{c|}{$40.64$} &
  $0.893$ &
  \multicolumn{1}{c|}{$63.04$} &
  $1.371$ &
  \multicolumn{1}{c|}{$46.08$} &
  $0.985$ &
  \multicolumn{1}{c|}{$62.08$} &
  $1.342$ &
  TRUE \\ \hline
$6$ &
  \multicolumn{1}{c|}{$8.96$} &
  $0.713$ &
  \multicolumn{1}{c|}{$16.0$} &
  $1.254$ &
  \multicolumn{1}{c|}{$9.44$} &
  $0.757$ &
  \multicolumn{1}{c|}{$8.96$} &
  $0.731$ &
  FALSE \\ \hline
$8$ &
  \multicolumn{1}{c|}{$121.6$} &
  $8.933$ &
  \multicolumn{1}{c|}{$207.68$} &
  $15.401$ &
  \multicolumn{1}{c|}{$124.8$} &
  $9.354$ &
  \multicolumn{1}{c|}{$170.24$} &
  $12.573$ &
  TRUE \\ \hline
$10$ &
  \multicolumn{1}{c|}{$8.0$} &
  $1.889$ &
  \multicolumn{1}{c|}{$16.0$} &
  $3.311$ &
  \multicolumn{1}{c|}{$8.0$} &
  $1.807$ &
  \multicolumn{1}{c|}{$8.0$} &
  $1.788$ &
  FALSE \\ \hline
\end{tabular}%
}
\caption{Performance of HyProVer for different Beta priors}
\label{diff_prior}
\end{table*}

%% file: table1.tex
\begin{table*}[]
	\resizebox{\textwidth}{!}{%
		\begin{tabular}{|c|c|c|cc|cc|cc|cc|}
			\hline
			\multirow{2}{*}{n} &
			\multirow{2}{*}{($\alpha$, $\beta$)} &
			\multirow{2}{*}{K} &
			\multicolumn{2}{c|}{$\toolname$} &
			\multicolumn{2}{c|}{SPRT ($\epsilon = 0.01$)} &
			\multicolumn{2}{c|}{SPRT ($\epsilon = 0.001$)} &
			\multicolumn{2}{c|}{Status} \\ \cline{4-11} 
			&
			&
			&
			\multicolumn{1}{c|}{Samples} &
			Time &
			\multicolumn{1}{c|}{Samples} &
			Time &
			\multicolumn{1}{c|}{Samples} &
			Time &
			\multicolumn{1}{c|}{$\toolname$} &
			SPRT \\ \hline\hline
			4 &
			$0.01$, $0.01$ &
			$3$ &
			\multicolumn{1}{c|}{$19.68$} &
			$0.019$ &
			\multicolumn{1}{c|}{$256.0$} &
			$0.147$ &
			\multicolumn{1}{c|}{$2048.0$} &
			$0.994$ &
			\multicolumn{1}{c|}{TRUE} &
			TRUE \\ \hline
			$6$ &
			$0.01$, $0.01$ &
			$3$ &
			\multicolumn{1}{c|}{$8.0$} &
			$0.049$ &
			\multicolumn{1}{c|}{-} &
			$>1800$ &
			\multicolumn{1}{c|}{-} &
			$>1800$ &
			\multicolumn{1}{c|}{TRUE} &
			UNDECIDED \\ \hline
			$8$ &
			$0.01$, $0.01$ &
			$3$ &
			\multicolumn{1}{c|}{$8.0$} &
			$0.136$ &
			\multicolumn{1}{c|}{-} &
			$>1800$ &
			\multicolumn{1}{c|}{-} &
			$>1800$ &
			\multicolumn{1}{c|}{TRUE} &
			UNDECIDED \\ \hline
			$10$ &
			$0.01$, $0.01$ &
			$3$ &
			\multicolumn{1}{c|}{$8.0$} &
			$0.309$ &
			\multicolumn{1}{c|}{-} &
			$>1800$ &
			\multicolumn{1}{c|}{-} &
			$>1800$ &
			\multicolumn{1}{c|}{TRUE} &
			UNDECIDED \\ \hline
			$4$ &
			$0.001$, $0.001$ &
			$3$ &
			\multicolumn{1}{c|}{$33.92$} &
			$0.027$ &
			\multicolumn{1}{c|}{$471.04$} &
			$0.253$ &
			\multicolumn{1}{c|}{$4096.0$} &
			$1.979$ &
			\multicolumn{1}{c|}{TRUE} &
			TRUE \\ \hline
			$6$ &
			$0.001$, $0.001$ &
			$3$ &
			\multicolumn{1}{c|}{$16.0$} &
			$0.059$ &
			\multicolumn{1}{c|}{-} &
			$>1800$ &
			\multicolumn{1}{c|}{-} &
			$>1800$ &
			\multicolumn{1}{c|}{TRUE} &
			UNDECIDED \\ \hline
			$8$ &
			$0.001$, $0.001$ &
			$3$ &
			\multicolumn{1}{c|}{$16.0$} &
			$0.151$ &
			\multicolumn{1}{c|}{-} &
			$>1800$ &
			\multicolumn{1}{c|}{-} &
			$>1800$ &
			\multicolumn{1}{c|}{TRUE} &
			UNDECIDED \\ \hline
			$10$ &
			$0.001$, $0.001$ &
			$3$ &
			\multicolumn{1}{c|}{$16.0$} &
			$0.329$ &
			\multicolumn{1}{c|}{-} &
			$>1800$ &
			\multicolumn{1}{c|}{-} &
			$>1800$ &
			\multicolumn{1}{c|}{TRUE} &
			UNDECIDED \\ \hline
			$4$ &
			$0.01$, $0.01$ &
			$8$ &
			\multicolumn{1}{c|}{$3817.6$} &
			$3.270$ &
			\multicolumn{1}{c|}{$3604.48$} &
			$3.784$ &
			\multicolumn{1}{c|}{$38666.24$} &
			$33.261$ &
			\multicolumn{1}{c|}{FALSE} &
			FALSE \\ \hline
			$6$ &
			$0.01$, $0.01$ &
			$8$ &
			\multicolumn{1}{c|}{$33.12$} &
			$0.117$ &
			\multicolumn{1}{c|}{$348.16$} &
			$0.849$ &
			\multicolumn{1}{c|}{$4096.0$} &
			$8.887$ &
			\multicolumn{1}{c|}{TRUE} &
			TRUE \\ \hline
			$8$ &
			$0.01$, $0.01$ &
			$8$ &
			\multicolumn{1}{c|}{$11.84$} &
			$0.169$ &
			\multicolumn{1}{c|}{$225.28$} &
			$2.054$ &
			\multicolumn{1}{c|}{$2048.0$} &
			$9.131$ &
			\multicolumn{1}{c|}{TRUE} &
			TRUE \\ \hline
			$10$ &
			$0.01$, $0.01$ &
			$8$ &
			\multicolumn{1}{c|}{$8.0$} &
			$0.336$ &
			\multicolumn{1}{c|}{-} &
			$>1800$ &
			\multicolumn{1}{c|}{-} &
			$>1800$ &
			\multicolumn{1}{c|}{TRUE} &
			UNDECIDED \\ \hline
			$4$ &
			$0.001$, $0.001$ &
			$8$ &
			\multicolumn{1}{c|}{$8785.92$} &
			$7.452$ &
			\multicolumn{1}{c|}{$6144.0$} &
			$6.203$ &
			\multicolumn{1}{c|}{$64880.64$} &
			$55.780$ &
			\multicolumn{1}{c|}{FALSE} &
			FALSE \\ \hline
			$6$ &
			$0.001$, $0.001$ &
			$8$ &
			\multicolumn{1}{c|}{$56.64$} &
			$0.164$ &
			\multicolumn{1}{c|}{$512.0$} &
			$1.195$ &
			\multicolumn{1}{c|}{$4096.0$} &
			$8.889$ &
			\multicolumn{1}{c|}{TRUE} &
			TRUE \\ \hline
			$8$ &
			$0.001$, $0.001$ &
			$8$ &
			\multicolumn{1}{c|}{$21.12$} &
			$0.204$ &
			\multicolumn{1}{c|}{$256.0$} &
			$2.310$ &
			\multicolumn{1}{c|}{$2048.0$} &
			$8.916$ &
			\multicolumn{1}{c|}{TRUE} &
			TRUE \\ \hline
			$10$ &
			$0.001$, $0.001$ &
			$8$ &
			\multicolumn{1}{c|}{$16.0$} &
			$0.382$ &
			\multicolumn{1}{c|}{-} &
			$>1800$ &
			\multicolumn{1}{c|}{-} &
			$>1800$ &
			\multicolumn{1}{c|}{TRUE} &
			UNDECIDED \\ \hline
		\end{tabular}%
	}
	\caption{SMC of collision avoidance formula}
	\label{col_avoid_table}
\end{table*}

%% file: table2.tex
\begin{table*}[]
\resizebox{\textwidth}{!}{%
\begin{tabular}{|c|c|c|cc|cc|cc|cc|}
\hline
\multirow{2}{*}{n} &
  \multirow{2}{*}{($\alpha$, $\beta$)} &
  \multirow{2}{*}{K} &
  \multicolumn{2}{c|}{$\toolname$} &
  \multicolumn{2}{c|}{SPRT ($\epsilon = 0.01$)} &
  \multicolumn{2}{c|}{SPRT ($\epsilon = 0.001$)} &
  \multicolumn{2}{c|}{Status} \\ \cline{4-11} 
 &
   &
   &
  \multicolumn{1}{c|}{Samples} &
  Time &
  \multicolumn{1}{c|}{Samples} &
  Time &
  \multicolumn{1}{c|}{Samples} &
  Time &
  \multicolumn{1}{c|}{$\toolname$} &
  SPRT \\ \hline\hline
$4$ &
  $0.01$, $0.01$ &
  $3$ &
  \multicolumn{1}{c|}{$40.32$} &
  $0.216$ &
  \multicolumn{1}{c|}{-} &
  $>1800$ &
  \multicolumn{1}{c|}{-} &
  $>1800$ &
  \multicolumn{1}{c|}{TRUE} &
  UNDECIDED \\ \hline
$6$ &
  $0.01$, $0.01$ &
  $3$ &
  \multicolumn{1}{c|}{$16.0$} &
  $0.278$ &
  \multicolumn{1}{c|}{-} &
  $>1800$ &
  \multicolumn{1}{c|}{-} &
  $>1800$ &
  \multicolumn{1}{c|}{FALSE} &
  UNDECIDED \\ \hline
$8$ &
  $0.01$, $0.01$ &
  $3$ &
  \multicolumn{1}{c|}{$27.68$} &
  $0.604$ &
  \multicolumn{1}{c|}{-} &
  $>1800$ &
  \multicolumn{1}{c|}{-} &
  $>1800$ &
  \multicolumn{1}{c|}{TRUE} &
  UNDECIDED \\ \hline
$10$ &
  $0.01$, $0.01$ &
  $3$ &
  \multicolumn{1}{c|}{$16.0$} &
  $0.951$ &
  \multicolumn{1}{c|}{-} &
  $>1800$ &
  \multicolumn{1}{c|}{-} &
  $>1800$ &
  \multicolumn{1}{c|}{FALSE} &
  UNDECIDED \\ \hline
$4$ &
  $0.001$, $0.001$ &
  $3$ &
  \multicolumn{1}{c|}{$57.28$} &
  $0.580$ &
  \multicolumn{1}{c|}{-} &
  $>1800$ &
  \multicolumn{1}{c|}{-} &
  $>1800$ &
  \multicolumn{1}{c|}{TRUE} &
  UNDECIDED \\ \hline
$6$ &
  $0.001$, $0.001$ &
  $3$ &
  \multicolumn{1}{c|}{$32.0$} &
  $0.935$ &
  \multicolumn{1}{c|}{-} &
  $>1800$ &
  \multicolumn{1}{c|}{-} &
  $>1800$ &
  \multicolumn{1}{c|}{FALSE} &
  UNDECIDED \\ \hline
$8$ &
  $0.001$, $0.001$ &
  $3$ &
  \multicolumn{1}{c|}{$62.72$} &
  $2.096$ &
  \multicolumn{1}{c|}{-} &
  $>1800$ &
  \multicolumn{1}{c|}{-} &
  $>1800$ &
  \multicolumn{1}{c|}{TRUE} &
  UNDECIDED \\ \hline
$10$ &
  $0.001$, $0.001$ &
  $3$ &
  \multicolumn{1}{c|}{$32.0$} &
  $2.593$ &
  \multicolumn{1}{c|}{-} &
  $>1800$ &
  \multicolumn{1}{c|}{-} &
  $>1800$ &
  \multicolumn{1}{c|}{FALSE} &
  UNDECIDED \\ \hline
$4$ &
  $0.01$, $0.01$ &
  $8$ &
  \multicolumn{1}{c|}{$34.24$} &
  $0.730$ &
  \multicolumn{1}{c|}{-} &
  $>1800$ &
  \multicolumn{1}{c|}{-} &
  $>1800$ &
  \multicolumn{1}{c|}{TRUE} &
  UNDECIDED \\ \hline
$6$ &
  $0.01$, $0.01$ &
  $8$ &
  \multicolumn{1}{c|}{$18.56$} &
  $1.439$ &
  \multicolumn{1}{c|}{-} &
  $>1800$ &
  \multicolumn{1}{c|}{-} &
  $>1800$ &
  \multicolumn{1}{c|}{FALSE} &
  UNDECIDED \\ \hline
$8$ &
  $0.01$, $0.01$ &
  $8$ &
  \multicolumn{1}{c|}{$33.76$} &
  $2.386$ &
  \multicolumn{1}{c|}{-} &
  $>1800$ &
  \multicolumn{1}{c|}{-} &
  $>1800$ &
  \multicolumn{1}{c|}{TRUE} &
  UNDECIDED \\ \hline
$10$ &
  $0.01$, $0.01$ &
  $8$ &
  \multicolumn{1}{c|}{$16.0$} &
  $3.318$ &
  \multicolumn{1}{c|}{-} &
  $>1800$ &
  \multicolumn{1}{c|}{-} &
  $>1800$ &
  \multicolumn{1}{c|}{FALSE} &
  UNDECIDED \\ \hline
$4$ &
  $0.001$, $0.001$ &
  $8$ &
  \multicolumn{1}{c|}{$53.76$} &
  $1.125$ &
  \multicolumn{1}{c|}{-} &
  $>1800$ &
  \multicolumn{1}{c|}{-} &
  $>1800$ &
  \multicolumn{1}{c|}{TRUE} &
  UNDECIDED \\ \hline
$6$ &
  $0.001$, $0.001$ &
  $8$ &
  \multicolumn{1}{c|}{$33.92$} &
  $2.646$ &
  \multicolumn{1}{c|}{-} &
  $>1800$ &
  \multicolumn{1}{c|}{-} &
  $>1800$ &
  \multicolumn{1}{c|}{FALSE} &
  UNDECIDED \\ \hline
$8$ &
  $0.001$, $0.001$ &
  $8$ &
  \multicolumn{1}{c|}{$53.44$} &
  $3.730$ &
  \multicolumn{1}{c|}{-} &
  $>1800$ &
  \multicolumn{1}{c|}{-} &
  $>1800$ &
  \multicolumn{1}{c|}{TRUE} &
  UNDECIDED \\ \hline
$10$ &
  $0.001$, $0.001$ &
  $8$ &
  \multicolumn{1}{c|}{$32.0$} &
  $6.316$ &
  \multicolumn{1}{c|}{-} &
  $>1800$ &
  \multicolumn{1}{c|}{-} &
  $>1800$ &
  \multicolumn{1}{c|}{FALSE} &
  UNDECIDED \\ \hline
\end{tabular}%
}
\caption{SMC of collision free goal reaching formula for first robot}
\label{gol_reach_table}
\end{table*}

%% file: conclusion.tex
\section{conclusion}\label{conc}
In this paper, we have developed a recursive statistical model checking algorithm for verifying discrete-time probabilistic hyperproperties on discrete-time Markov chains. Our broad approach consisted of mapping the $\hyppctl$ verification problem to an $n$-dimensional hypotheses testing problem. We designed an algorithm based on random sampling followed by Bayes' test for non-nested $\hyppctl$ formula, and extended it to the nested cases through a recursive algorithm that exploits an approximate Bayes' test. Finally, we used our algorithm to verify probabilistic hyperproperties like collision avoidance and collision free goal reaching on the grid world robot navigation scenarios. We compared the performance of our algorithm against the $\sprt$ based algorithm discussed in \cite{wang2021statistical} and showed that our algorithm performs better both in verification time and number of required samples for inference; a stark difference arises when we consider nested formulae.

For future work, we would like to develop a verification algorithm based on Bayes' test for hyperproperties over continuous time. Another interesting research direction would be to incorporate unbounded temporal (until) operators. This would enable the verification of unbounded probabilistic hyperproperties using light-weight methods such as Bayesian $\smc$.